\documentclass[final,12pt]{article}
\usepackage{amsfonts,color,morefloats,pslatex,a4wide}
\usepackage{amsmath, amssymb, mathrsfs}
\usepackage{mathtools}
\usepackage{amsthm}
\usepackage[utf8]{inputenc}
\usepackage{bm}
\usepackage{array}
\usepackage{booktabs}
\usepackage{lscape}
\usepackage{graphicx}
\usepackage{amssymb}
\usepackage{longtable}
\usepackage{threeparttable}
\usepackage{adjustbox}
\usepackage{url}
\usepackage{makecell}
\usepackage{caption}
\usepackage{amsmath,bm}
\usepackage[compress]{cite}

\newtheorem{maintheorem}{Theorem}[section]

\newtheorem{theorem}[maintheorem]{Theorem}
\newtheorem{lemma}[maintheorem]{Lemma}

\newtheorem{definition}[maintheorem]{Definition}
\newtheorem{example}[maintheorem]{Example}
\newtheorem{remark}{Remark}

\usepackage[hidelinks]{hyperref}

\usepackage{caption}
\usepackage{booktabs}
\usepackage{float}
\usepackage{enumitem}

\DeclareMathOperator{\Hull}{Hull}
\DeclareMathOperator{\Sum}{Sum}
\DeclareMathOperator{\lcm}{lcm}
\DeclareMathOperator{\dimension}{dim}

\makeatletter
\renewcommand{\thanks}[1]{\g@addto@macro\@thanks{\footnotetext{#1}}}

\newcommand{\Rmnum}[1]{\expandafter\@slowromancap\romannumeral #1@}
\makeatother

\title{Hulls and sums of separable constacyclic codes over $\mathbb{F}_q \times (\mathbb{F}_q+v\mathbb{F}_q)$ and new quantum codes}

\author{
	Yu Qian\textsuperscript{1}, 
	Yu Wang\textsuperscript{1,$\ddagger$}, 
	Liqi Wang\textsuperscript{2}
	\thanks{
	\textsuperscript{1}School of Mathematics and Statistics, Hefei University, Hefei 230601, China,\\
	\textsuperscript{2}School of Mathematics, Hefei University of Technology, Hefei 230009, China.\\
		Emails: 13856012927@163.com, 
		\textsuperscript{$\ddagger$}wangyu351@hfuu.edu.cn (Corresponding author), 
		liqiwangg@163.com.
	}
	\date{}
}

\begin{document}
	
	\maketitle
	
	\begin{abstract}
		\vspace*{.3cm}
		\noindent
		We establish the generator polynomials of the Euclidean and Hermitian duals of separable constacyclic codes over $\mathcal{S} = \mathbb{F}_q \times (\mathbb{F}_q+v\mathbb{F}_q)$, with $q$ an odd prime power and $v^2=v$, and we derive the generator polynomials of their Gray images, respectively. The generator polynomials of the Euclidean hulls and Hermitian hulls of separable constacyclic codes over $\mathcal{S}$ and their Gray images are presented, respectively. Furthermore, we provide the generator polynomials of the Euclidean sums and Hermitian sums of separable constacyclic codes and their Gray images, respectively. Finally, we propose two methods to yield quantum error-correcting codes (QECCs) from the hulls and sums of separable constacyclic codes over $\mathcal{S}$, and generate new QECCs that outperform the existing ones in terms of parameters.
		

	\end{abstract}
	
	\section{Introduction}\
	
	Quantum error correction is a key technology in quantum computing that protects quantum information from being disrupted by noise. Its core principle lies in using redundant qubits to store information and enhancing system stability through error detection and correction procedures. There are some famous methods to construct QECCs from classical error-correcting codes over finite field, such as: Calderbank–Shor–Steane (CSS) construction \cite{Calderbank1998}, Hermitian construction \cite{Ashikhmin2001} and quantum construction X \cite{Lison2014}.
	
	In recent years, an increasing number of scholars have focused their attention on  constacyclic codes over finite non-chain rings and utilized them to generate QECCs. In 2018, Gao et al. \cite{Gao2018} employed $u$-constacyclic codes over $\mathbb{F}_p+ u\mathbb{F}_p(u^2=1)$ to obtain several novel QECCs by the CSS construction. Simultaneously,  
	Ma et al. \cite{Ma2018} characterized the basic framework of $(\alpha+\beta v+\gamma v^2)$-constacyclic codes over $\mathbb{F}_q+v\mathbb{F}_q+v^2\mathbb{F}_q(v^3=v)$, and produced some new non-binary QECCs. Gao et al. \cite{Gao2019} utilized cyclic codes over $\mathbb{F}_q+v_1\mathbb{F}_q+\cdots+v_r\mathbb{F}_q(v_{i}^2=v_{i},v_{i}v_{j}=v_{j}v_{i}=0)$ to yield previously unknown $q$-ary QECCs. Wang et al. \cite{Wang2021} derived QECCs  from constacyclic codes over $\mathbb{F}_{q^2}+v\mathbb{F}_{q^2}(v^2=v)$ with the help of the Hermitian construction, and they got many QECCs superior to the known ones. Further, Gao et al. \cite{Gao2023} proposed an extension of the construction of  QECCs to entanglement-assisted quantum error-correcting codes (EAQECCs) over the same ring in \cite{Wang2021} and generated some new maximum distance separable (MDS) EAQECCs. Recently, some scholars have started to investigate the hulls and sums of constacyclic codes. Tian et al. \cite{Tian2024} produced many well-performing QECCs via the Euclidean and Hermitian hulls of constacyclic codes over the same ring in \cite{Wang2021}. Next, Pandey et al. \cite{Pandey2024} used the Euclidean hulls and sums of cyclic codes over the same ring in \cite{Gao2019} to generate QECCs and EAQECCs. Moreover, Yadav et al. \cite{Yadav2024} introduced the Hermitian hull of constacyclic codes over $\mathbb{F}_{q^2}+u\mathbb{F}_{q^2}+u^2\mathbb{F}_{q^2}+\cdots+u^{e-1}\mathbb{F}_{q^2}(u^e=1)$ to construct several optimal QECCs. Very recently, Zhang et al. \cite{Zhang2025} investigated three methods for generating QECCs from cyclic codes over $\mathbb{F}_{l}[\gamma]/\langle \gamma^2 \rangle$, i.e., CSS construction and the quantum construction X of the Euclidean and Hermitian duals, and produced a set of QECCs. The constructions of QECCs have seen notable progress (see \cite{Wang2023,Zhang2023,Li2025} and related works).
	
	Meanwhile, researchers have attempted to generalize the investigation of contacyclic codes over the single ring to the mixed alphabets, and tried to use them  to construct QECCs and EAQECCs. In 2021, Dinh et al. \cite{Dinh2021-2} studied the fundamental properties of $\mathbb{F}_q\mathcal{R}\mathcal{S}$-constacyclic codes, where $\mathcal{R} = \mathbb{F}_q + u\mathbb{F}_q(u^2=u)$ and $\mathcal{S} = \mathbb{F}_q + u\mathbb{F}_q + v\mathbb{F}_q(u^2=u, v^2=v, uv=vu=0)$. Then they produced some QECCs by the CSS construction. Later, {\c{C}}al{\i}{\c{s}}kan et al. \cite{Caliskan2023-1} discussed the Euclidean dual-containing codes over $\mathbb{F}_p \times (\mathbb{F}_p + v\mathbb{F}_p)(v^2=v)$. Then they listed a large number of new QECCs. Recently, Wang et al. \cite{Wang2025} introduced the basic framework of separable $\mathbb{F}_{q^2}\mathcal{R}$-$u$-constacyclic codes, and obtained a series of QECCs with better parameters via the Hermitian construction, where $\mathcal{R} = \mathbb{F}_{q^2} + v\mathbb{F}_{q^2}(v^2=v)$. After that, Qian et al. \cite{Qian2026} generalized $\mathcal{R}$ in \cite{Wang2025} to $R_r$, where $R_r = \mathbb{F}_{q} + v_{1}\mathbb{F}_{q} + \cdots + v_{r}\mathbb{F}_{q}(v_{i}^2 = v_{i}, v_{i}v_{j} = v_{j}v_{i}, 1 \leq i \neq j \leq r)$, and then they established two methods to produce numerous high-performance QECCs stem from separable constacyclic codes over $\mathbb{F}_{q}R_r$, including the CSS and Hermitian methods. Next, regarding the quantum construction X method over mixed alphabets, Hu et al. \cite{Hu2025} generated some good QECCs from cyclic codes over $\mathbb{F}_2 \times (\mathbb{F}_2 + v\mathbb{F}_2)$, where $v^2=v$. Subsequently, Shukla et al. \cite{Shukla2026} investigated the hulls of separable cyclic codes over $\mathbb{Z}_p\mathbb{Z}_p[v]$, where $p$ is an arbitrary prime and $v^2=v$. They used the hulls of non-separable cyclic codes over $\mathbb{Z}_2\mathbb{Z}_2[v]$ to generate EAQECCs with superior parameters. Further research can be found in \cite{Li2020,Dinh2021-1,Ashraf2022,Ali2025, Prakash2025}.
	
	Motived by the important conclusions and innovative methods in \cite{Tian2024,Wang2025,Hu2025}, we study the generator polynomials of the Euclidean and Hermitian hulls and sums of separable constacyclic codes and their Gray images over $\mathcal{S} = \mathbb{F}_q \times (\mathbb{F}_q + v\mathbb{F}_q)(v^2=v)$, and aim to derive new QECCs from the Euclidean and Hermitian hulls and sums of separable constacyclic codes over such ring.  
	
	The remainder of this paper is arranged as follows: Sect. 2 recalls the fundamental framework of constacyclic codes over $\mathbb{F}_q + v\mathbb{F}_q$ and $\mathcal{S}$. Sect. 3 first studies the Euclidean and Hermitian duals of separable $\mathcal{S}$-$t$-constacyclic codes and their Gray images, where $t=(\lambda, \beta(1-2v))$. Then the generator polynomials of the hulls and their Gray images are discussed under the Euclidean and Hermitian inner products. Further the generator polynomials of the sums and their Gray images are also provided. In Sect. 4, as an application, some QECCs with superior parameters are presented from the Euclidean and Hermitian hulls and sums of separable constacyclic codes over $\mathcal{S}$, respectively. Sect. 5 presents a summary of the whole paper.

	\section{Preliminaries}\
	
	\subsection{Constacyclic codes over $\mathcal{R}$}
	
	Suppose that \(\mathcal{R} = \mathbb{F}_q+v\mathbb{F}_q\), with \(q\) an odd prime power and \(v^2 = v\).
	Given that $\mathscr{C}_n$ is a $\mathcal{R}$-linear code of length $n$. Let
	
	\[
	C_{n,1} = \left\{ \boldsymbol{\ell}_1 \in \mathbb{F}_{q}^n \mid \exists \boldsymbol{\ell}_2 \in \mathbb{F}_{q}^n, \, (1 - v)\boldsymbol{\ell}_1 + v\boldsymbol{\ell}_2 \in \mathscr{C}_n \right\},
	\]
	
	\[
	C_{n,2} = \left\{ \boldsymbol{\ell}_2 \in \mathbb{F}_{q}^n \mid \exists \bm{\ell}_1 \in \mathbb{F}_{q}^n, \, (1 - v)\boldsymbol{\ell}_1 + v\boldsymbol{\ell}_2 \in \mathscr{C}_n \right\}.
	\]
	Thus, $C_{n,1}$ and $C_{n,2}$ are linear codes of length $n$ over $\mathbb{F}_{q}$. And
	
	\[
	\mathscr{C}_n = (1 - v)C_{n,1} \oplus vC_{n,2}.
	\]
	Therefore,
	$|\mathscr{C}_n| = |C_{n,1}||C_{n,2}|.$
	
	Denote by \( GL_2(\mathbb{F}_{q}) \) the collection of all \( 2 \times 2 \) invertible matrices over \( \mathbb{F}_{q} \).
	A Gray map from \( \mathcal{R} \) to \( \mathbb{F}_{q}^{2} \) is given as follows:
	
	\[
	\phi_M : \mathcal{R} \to \mathbb{F}_{q}^2,
	\]
	\[
	\gamma = (1 - v)\gamma_1 + v\gamma_2 \mapsto (\gamma_1, \gamma_2)M = \boldsymbol{\gamma}M,
	\]
	where \( M \in GL_2(\mathbb{F}_{q}) \). For a given matrix \( M \),
	the Gray map \( \phi_M \) is obviously bijective
	and can be generalized naturally to \( \mathcal{R}^n \):
	
	\[
	\phi_M : \mathcal{R}^n \to \mathbb{F}_{q}^{2n},
	\]
	\[
	\boldsymbol{\xi} = (\xi_0, \xi_1, \ldots, \xi_{n-1}) \mapsto (\phi_M({\xi}_0), \phi_M({\xi}_1), \ldots, \phi_M({\xi}_{n-1})).
	\]

	\begin{definition}
		Suppose $\mathscr{C}_n$ is a $\mathcal{R}$-linear code of length $n$ and $\mu$ is a unit of $\mathcal{R}$. If for any codeword \(\mathbf{c} = (c_0, c_1, \ldots, c_{n-1}) \in \mathscr{C}_n\) implying
		$(\mu c_{n-1}, c_0, c_1, \ldots, c_{n-2}) \in \mathscr{C}_n$,
		thus \(\mathscr{C}_n\) is a $\mathcal{R}$-\(\mu\)-constacyclic code. A $\mathcal{R}$-\(\mu\)-constacyclic code with \(\mu = 1\) is a cyclic code; while for \(\mu = -1\), it is a negacyclic code.
	\end{definition}
	
	Now, consider $\mu = \beta(1-2v)$, where $\beta \in \mathbb{F}_{q}\setminus \{0\}$. Similar to the discussion in \cite{Wang2021}, the following conclusions are obtained.
	
	\begin{lemma}
		Assume that $\mathscr{C}_n =(1 - v)C_{n,1} \oplus vC_{n,2}$ is a $\mathcal{R}$-$\mu$-constacyclic code of length $n$. Thus,
		
		\begin{enumerate}
			\item[(i)] $\mathscr{C}_n$ is a $\mathcal{R}$-$\mu$-constacyclic code if and only if $C_{n,1}$ and $C_{n,2}$ are $\beta$-constacyclic and $(-\beta)$-constacyclic codes of length $n$ over $\mathbb{F}_{q}$, respectively.
			
			\item[(ii)] There is a unique polynomial $g(x) = (1 - v)g_1(x) + vg_2(x)$ that satisfies $g(x) \mid (x^n - \mu)$ and for which $\mathscr{C}_n = \langle g(x) \rangle$, where $C_{n,1} = \langle g_1(x) \rangle $ and $C_{n,2} = \langle g_2(x) \rangle $. Moreover, $|\mathscr{C}_n| = q^{2n - \deg(g_1(x)) - \deg(g_2(x))}$.
			
			\item[(iii)] Take $M = \begin{pmatrix} \vartheta_1 \beta & \vartheta_1 \\ - \vartheta_2 \beta & \vartheta_2 \end{pmatrix}$, where $\vartheta_1, \vartheta_2 \in \mathbb{F}_{q}\setminus \{0\}$. Then $\phi_M(\mathscr{C}_n) = \langle g_1(x)g_2(x) \rangle$ is a $\beta^2$-constacyclic code of length $2n$ over $\mathbb{F}_{q}$.
			
		\end{enumerate}
	\end{lemma}

	\subsection{Constacyclic codes over $\mathcal{S}$}

	Let $\mathcal{S}=\mathbb{F}_q\mathcal{R}=\mathbb{F}_q \times (\mathbb{F}_q+v\mathbb{F}_q) = \left\{ (a, b) \ \middle|\ a \in \mathbb{F}_q, b \in \mathbb{F}_q+v\mathbb{F}_q \right\}$. 
	
	\begin{definition}
		Assume that $\mathcal{C}$ is a non-empty subset of $\mathbb{F}_{q}^m \times \mathcal{R}^n$. If $\mathcal{C}$ is a $\mathcal{R}$-submodule, then $\mathcal{C}$ is an $\mathcal{S}$-linear code of length $(m,n)$.
	\end{definition}

	\begin{definition}
		Suppose $\mathcal{C}$ is an $\mathcal{S}$-linear code of length $(m, n)$. Define $\mathrm{C}_m = \tau_m(\mathcal{C})$ and $\mathscr{C}_n = \tau_n(\mathcal{C})$, where the maps $\tau_m$ and $\tau_n$ are given by $\tau_m(\boldsymbol{\eta}, \boldsymbol{\xi}) = \boldsymbol{\eta}$ and $\tau_n(\boldsymbol{\eta}, \boldsymbol{\xi}) = \boldsymbol{\xi}$. Then $\mathrm{C}_m$ is a linear code of length $m$ over $\mathbb{F}_q$ and $\mathscr{C}_n$ is a $\mathcal{R}$-linear code of length $n$. The code $\mathcal{C}$ is separable if $\mathcal{C} = \mathrm{C}_m \times \mathscr{C}_n$. In this case, $|\mathcal{C}| = |\mathrm{C}_m| \cdot |\mathscr{C}_n|$.
		
	\end{definition}

	\begin{definition}
		Suppose $\mathcal{C}$ is an $\mathcal{S}$-linear code of length $(m,n)$ and $t=(\lambda,\mu)$ is a unit of $\mathcal{S}$, where $\lambda$ and $\mu$ are units of $\mathbb{F}_{q}$ and $\mathcal{R}$, respectively. If for any codeword \(\mathbf{c} = (c_{1,0}, c_{1,1}, \ldots, c_{1,m-1}, c_{2,0},c_{2,1}, 
		\\\ldots, c_{2,n-1}) \in \mathcal{C}\) implying
		$(\lambda c_{1,m-1}, c_{1,0}, \ldots, c_{1,m-2}, \mu c_{2,n-1}, c_{2,0}, \ldots, c_{2,n-2}) \in \mathcal{C}$, thus \(\mathcal{C}\) is an $\mathcal{S}$-\(t\)-constacyclic code.
	\end{definition}

	From \cite{Qian2026}, let $r=1$ in $\mathbb{F}_{q}R_r$, then $\mathcal{S} = \mathbb{F}_{q}R_1$, $e_0 = 1-v$ and $e_1 = v$. Further, suppose $\xi_{0}=\beta$ and $\xi_{1}=-\beta$, thus $t=(\lambda, \beta(1-2v))$. In what follows, assume that $t=(\lambda,\mu)$, where $\mu = \beta(1-2v)$ and $\lambda, \beta \in \mathbb{F}_{q}\setminus \{0\}$. Therefore,  the following lemmas are all obviously true.

	\begin{lemma}(\cite{Qian2026})
		Suppose \( \mathcal{C} \) is an $\mathcal{S}$-$t$-constacyclic code of length $(m, n)$. Thus,
		\[ \mathcal{C} = \langle (g_0(x),0), (y(x),(1 - v)g_1(x) + vg_2(x)) \rangle, \] where $g_0(x), y(x)\in \mathbb{F}_{q}[x]/\langle x^m - \lambda \rangle$, $g_i(x) \in \mathbb{F}_q[x]$, $i=1,2$, and \( g_0(x) \mid (x^m - \lambda)\), \(g_1(x) \mid (x^n - \beta) \), \(g_2(x) \mid (x^n + \beta) \).	
	\end{lemma}

	\begin{lemma}(\cite{Qian2026})
		Suppose $\mathcal{C}$ is a separable $\mathcal{S}$-$t$-constacyclic code of length $(m, n)$. Thus $C_m$ is a $\lambda$-constacyclic code of length $m$ over $\mathbb{F}_{q}$ and $\mathscr{C}_n$ is a $\mathcal{R}$-\(\mu\)-constacyclic code of length $n$.
		
	\end{lemma}

	\begin{lemma}(\cite{Qian2026}) 
		Suppose $\mathcal{C} = \langle (g_0(x), 0), (y(x), (1 - v)g_1(x) + vg_2(x)) \rangle$ is a separable $\mathcal{S}$-$t$-constacyclic code. Then the following equivalences hold:
		
		\begin{enumerate}
			
			\item[(i)] $g_0(x)|y(x);$
			
			\item[(ii)] $y(x)=0;$
			
			\item[(iii)] $\mathcal{C} = \langle (g_0(x), 0), (0, (1 - v)g_1(x) + vg_2(x)) \rangle;$
			
			\item[(iv)] $C_m = \langle g_0(x) \rangle$, $\mathscr{C}_n = \langle (1 - v)g_1(x) + vg_2(x) \rangle.$
			
		\end{enumerate}
		
	\end{lemma}

	A Gray map \( \Phi_M \) from \( \mathcal{S} \) to \( \mathbb{F}_{q}^{3} \) is now defined as follows:
	
	\[
	\Phi_M : \mathcal{S} \to \mathbb{F}_{q}^{3},
	\]
	\[
	(\eta, \xi) \mapsto (\eta, \phi_M(\xi)),
	\]
	where \( M \in GL_2(\mathbb{F}_{q}) \).
	
	Let \( \boldsymbol{\eta} = (\eta_0, \eta_1, \dots, \eta_{m-1}) \in \mathbb{F}_{q}^m \), \( \boldsymbol{\xi} = (\xi_0, \xi_1, \dots, \xi_{n-1}) \in \mathcal{R}^n \). The Gray map \( \Phi_M \) can also be extended to \( \mathbb{F}_{q}^m \times \mathcal{R}^n \) naturally as follows:
	
	\[
	\Phi_M : \mathbb{F}_{q}^m \times \mathcal{R}^n \to \mathbb{F}_{q}^{m+2n},
	\]
	\[
	(\eta_0, \eta_1, \dots, \eta_{m-1}, \xi_0, \xi_1, \dots, \xi_{n-1}) \mapsto (\eta_0, \eta_1, \dots, \eta_{m-1}, \phi_M({\xi}_0), \phi_M({\xi}_1), \dots, \phi_M({\xi}_{n-1})).
	\]
	That is, $\Phi_M(\boldsymbol{\eta}, \boldsymbol{\xi}) = (\boldsymbol{\eta}, \phi_M(\boldsymbol{\xi}))$. From the linearity of $\phi_M$, it follows readily that $\Phi_M$ is also linear.
	
	For any element \( (\boldsymbol{\eta}, \boldsymbol{\xi}) \in \mathbb{F}_{q}^m \times \mathcal{R}^n \), define its Lee weight as 
	\[ w_L(\boldsymbol{\eta}, \boldsymbol{\xi}) = w_H(\boldsymbol{\eta}) + w_L(\boldsymbol{\xi}) = w_H(\boldsymbol{\eta}) + w_H(\phi_M(\boldsymbol{\xi})). \]
	 
	Additionally, the Lee distance between any two elements \( \boldsymbol{\varepsilon}_1, \boldsymbol{\varepsilon}_2 \) of \( \mathbb{F}_{q}^m \times \mathcal{R}^n \) is defined as
	\[ d_L(\boldsymbol{\varepsilon}_1, \boldsymbol{\varepsilon}_2) = w_L(\boldsymbol{\varepsilon}_1 - \boldsymbol{\varepsilon}_2) = w_H(\Phi_M(\boldsymbol{\varepsilon}_1 - \boldsymbol{\varepsilon}_2)). \]

	\begin{lemma}
		Suppose $\mathcal{C} = C_m \times \mathscr{C}_n$ is a separable $\mathcal{S}$-linear code of length $(m,n)$. Then $\Phi_M(\mathcal{C}) = C_m \times \phi_M(\mathscr{C}_n)$. In addition, $d(\Phi_M(\mathcal{C})) = \min\{d(C_m), d(\phi_M(\mathscr{C}_n))\}$.
	\end{lemma}
	
	\begin{proof}
		For any $\mathbf{c} \in \Phi_M(\mathcal{C})$, then there exists a $(\mathbf{c}_1, \mathbf{c}_2) \in \mathcal{C} = C_m \times \mathscr{C}_n$ such that $\mathbf{c} = \Phi_M(\mathbf{c}_1, \mathbf{c}_2)$. By the definition of $\Phi_M$:
		\[\mathbf{c} = (\mathbf{c}_1, \phi_M(\mathbf{c}_2)),\]
		where $\mathbf{c}_1 \in C_m$, $\mathbf{c}_2 \in \mathscr{C}_n$, so $\phi_M(\mathbf{c}_2) \in \phi_M(\mathscr{C}_n)$.
		Hence, $\mathbf{c} \in C_m \times \phi_M(\mathscr{C}_n)$, i.e., $\Phi_M(\mathcal{C}) \subset C_m \times \phi_M(\mathscr{C}_n)$.
		
		On the other hand, for every $\mathbf{c} = (\mathbf{c}_1, \mathbf{c}_2) \in C_m \times \phi_M(\mathscr{C}_n)$, then
		$\mathbf{c}_1 \in C_m$ and there exists an $\mathbf{a} \in \mathscr{C}_n$ such that $\mathbf{c}_2= \phi_M(\mathbf{a})$.
		Since $\mathcal{C} = C_m \times \mathscr{C}_n$, $(\mathbf{c}_1, \mathbf{a}) \in \mathcal{C}$. By $\Phi_M$,
		\[\Phi_M(\mathbf{c}_1, \mathbf{a}) = (\mathbf{c}_1, \phi_M(\mathbf{a})) = (\mathbf{c}_1, \mathbf{c}_2) = \mathbf{c}.\]
		Therefore, $\mathbf{c} \in \Phi_M(\mathcal{C})$, i.e., $C_m \times \phi_M(\mathscr{C}_n) \subset \Phi_M(\mathcal{C})$.
		
		In conclusion, $\Phi_M(\mathcal{C}) = C_m \times \phi_M(\mathscr{C}_n)$.
		Moreover, one can easily see that \[d(\Phi_M(\mathcal{C})) = \min\{d(C_m), d(\phi_M(\mathscr{C}_n))\}.\]
	\end{proof}

	\begin{lemma}
		Suppose $\mathcal{C} = \langle (g_0(x), 0), (0, (1 - v)g_1(x) + vg_2(x)) \rangle$ is a separable $\mathcal{S}$-$t$-constacyclic code of length $(m, n)$ and $M = \begin{pmatrix} \vartheta_1 \beta & \vartheta_1 \\ - \vartheta_2 \beta & \vartheta_2 \end{pmatrix}$, $\beta, \vartheta_1, \vartheta_2 \in \mathbb{F}_{q}\setminus \{0\}$. Then 
		\begin{enumerate}
			
			\item[(i)] \( \Phi_M(\mathcal{C}) = \langle (g_0(x),0), (0,g_1(x)g_2(x)) \rangle \) is a $(\lambda,\beta^2)$-constacyclic code of length $m+2n$ over $\mathbb{F}_{q};$
			
			\item[(ii)] $C_m = \langle g_0(x) \rangle$, $\phi_M(\mathscr{C}_n) = \langle g_1(x)g_2(x) \rangle;$
			
			\item[(iii)] $|\Phi_M(\mathcal{C})| = q^{m + 2n - \deg(g_0(x)) -  \deg(g_1(x)) - \deg(g_2(x))}.$
			
		\end{enumerate}

	\end{lemma}
	
	\begin{proof}
		(i) As in the proof of Theorem 3.15 in \cite{Wang2025}, \( \Phi_M(\mathcal{C}) = \langle (g_0(x),0), (0,g_1(x)g_2(x)) \rangle \) is a $(\lambda,\beta^2)$-constacyclic code of length $m+2n$ over $\mathbb{F}_{q}$. 
		
		(ii) This conclusion is immediate from Lemmas 2.8 and 2.2.

		(iii) Since $|\Phi_M(\mathcal{C})| = |\mathrm{C}_m| \cdot |\phi_M(\mathscr{C}_n)|$, and the bijectivity of $\phi_M$,
		\[|\Phi_M(\mathcal{C})| = |\mathrm{C}_m| \cdot |\mathscr{C}_n| = q^{m + 2n - \deg(g_0(x)) -  \deg(g_1(x)) - \deg(g_2(x))}.\]
	\end{proof}

	\section{Hulls and sums of $\mathcal{S}$-$t$-constacyclic codes}
	
	Here, suppose $t=(\lambda,\mu)$, where $\mu = \beta(1-2v)$ and $\lambda, \beta \in \mathbb{F}_{q}\setminus \{0\}$.
	
	\subsection{Euclidean and Hermitian dual codes and their Gray images}
	
	Let $q=e^l$ with $e$ an odd prime and $l > 0$ an integer.
	
	\begin{definition}({\cite{Fan2017}})
		For any two elements
		\( \mathbf{x} = (x_0, x_1, \ldots, x_{m-1}) \),
		\( \mathbf{y} = (y_0, y_1, \ldots, y_{m-1}) \) in \( \mathbb{F}_q^m \),
		and for any integer \( k \), \( 0 \leq k < l \), their \( k \)-Galois inner product is:
		\[
		\langle \mathbf{x}, \mathbf{y} \rangle_k = \sum_{i=0}^{m-1} x_i y_i^{e^k}.
		\]
		
	\end{definition}
	
	Consider a linear code \( C \) of length \( m \) over \( \mathbb{F}_q \), its \( k \)-Galois dual code is given by:
	\[
	C^{\perp_k} = \left\{ (x_0, x_1, \dots, x_{m-1}) \in \mathbb{F}_q^m \ \middle|\ \sum_{i=0}^{m-1} x_i c_i^{e^k} = 0, \forall \, (c_0, c_1, \dots, c_{m-1}) \in C \right\}.
	\]
	
	In fact, for $k = 0$, the $k$-Galois inner product is the Euclidean inner product, while for $k = \frac{l}{2}$ ($l$ even), it is the Hermitian inner product.
	
	Specifically, the Euclidean dual of \(C\) over $\mathbb{F}_{q}$ is as follows:
	\[
	C^{\perp_E} = \left\{ (x_0, x_1, \ldots, x_{m-1}) \in \mathbb{F}_q^m \ \middle|\ \sum_{i=0}^{m-1} x_i c_i = 0, \forall \, (c_0, c_1, \ldots, c_{m-1}) \in C \right\}.
	\]

	Define the Hermitian dual of \(C\) over \(\mathbb{F}_{q}\) as:
	\[
	C^{\perp_H} = \left\{ (x_0, x_1, \ldots, x_{m-1}) \in \mathbb{F}_{q}^m \ \middle|\ \sum_{i=0}^{m-1} x_i c_i^p = 0, \forall \, (c_0, c_1, \ldots, c_{m-1}) \in C \right\},
	\]
	where $p=e^k$ and $q=e^l=e^{2k}=p^2$.

	In what follows, their properties over $\mathcal{R}$ and $\mathcal{S}$ are discussed, respectively.
	
	According to Lemma 4.2 and Theorem 5.2 in \cite{Qian2026}, the following two lemmas hold.

	\begin{lemma}(\cite{Qian2026})
		Suppose \( \mathcal{C}^{\perp_E} \) is the Euclidean dual of an $\mathcal{S}$-linear code \( \mathcal{C} \). Then
		$\Phi_M(\mathcal{C}^{\perp_E}) = \Phi_M(\mathcal{C})^{\perp_E},$
		where $MM^\top=diag(\theta_1, \theta_2)$, $M^\top$ is the transpose matrix of $M$ and $\theta_i \in  \mathbb{F}_{q}\setminus \{0\}$, $i=1,2$.
	\end{lemma}
	
	\begin{lemma}(\cite{Qian2026})
		Suppose \( \mathcal{C}^{\perp_H} \) is the Hermitian dual of an $\mathcal{S}$-linear code \( \mathcal{C} \). Then
		$\Phi_M(\mathcal{C}^{\perp_H}) = \Phi_M(\mathcal{C})^{\perp_H},$
		where $MM^\dagger=diag(\theta_1, \theta_2)$, $M^\dagger$ is the conjugate transpose matrix of $M$ and $\theta_i \in \mathbb{F}_{p^2}\setminus \{0\}$, $i=1,2$.
	\end{lemma}

	\begin{lemma}
		Let $\mathscr{C}_{n} = \langle (1 - v)g_1(x) + vg_2(x) \rangle$ be a $\mathcal{R}$-$\mu$-constacyclic code. Then
		\begin{enumerate}[label=(\arabic*)]
			
			\item[(i)] \( \mathscr{C}_{n}^{\perp_E} = \langle (1 - v)h_1^*(x) + vh_2^*(x) \rangle \) is a $\mathcal{R}$-$\mu^{-1}$-constacyclic code, where $\mu^{-1}=\beta^{-1}(1-2v)$, 
			\[  h_1(x) = \frac{x^n - \beta}{g_1(x)}, h_2(x) = \frac{x^n + \beta}{g_2(x)} \] in $\mathbb{F}_q[x]$ and $h_{i}^*(x) = h_{i}(0)^{-1} x^{\deg h_{i}(x)} h_{i}\left(\frac{1}{x}\right)$, $i=1,2$.
			Moreover,
			\(|\mathscr{C}_{n}^{\perp_E}| = q^{{\deg(g_1(x)) + \deg(g_2(x))}}.\)
			
			\item[(ii)] \( \mathscr{C}_{n}^{\perp_H} = \langle (1 - v)h_1^\dagger(x) + vh_2^\dagger(x) \rangle \) is a $\mathcal{R}$-$\mu^{-p}$-constacyclic code, where $\mu^{-p}=\beta^{-p}(1-2v)$, 
			\[ h_1(x) = \frac{x^n - \beta}{g_1(x)}, h_2(x) = \frac{x^n + \beta}{g_2(x)}\] in $\mathbb{F}_{p^2}[x]$, $h_{i}^\dagger(x) = h_{i}(0)^{-p} x^{\deg h_{i}(x)} h_{i}^{p}\left(\frac{1}{x}\right)$, $i=1,2$.  
			Moreover,
			\(|\mathscr{C}_{n}^{\perp_H}| = p^{2({\deg(g_1(x)) + \deg(g_2(x))})}.\)
			
		\end{enumerate}
	\end{lemma}
	
	\begin{proof}
		(i) From Lemma 2.8, $\mathscr{C}_n = \langle (1 - v)g_1(x) + vg_2(x) \rangle$. Then similar to Corollary 4.8 in \cite{Zhu2011}, $\mathscr{C}_{n}^{\perp_E} = \langle (1 - v)h_1^*(x) + vh_2^*(x)) \rangle$, where \[h_1(x) = \frac{x^n - \beta}{g_1(x)}, h_2(x) = \frac{x^n + \beta}{g_2(x)} \] in $\mathbb{F}_q[x]$, $h_{i}^*(x) = h_{i}(0)^{-1} x^{\deg h_{i}(x)} h_{i}\left(\frac{1}{x}\right)$, $i=1,2$, and \(|\mathscr{C}_{n}^{\perp_E}| = q^{{\deg(g_1(x)) + \deg(g_2(x))}}.\)
		
		(ii) For the Hermitian dual code, it can be proved similarly.
		\end{proof}

	\begin{theorem}
		Assume that $\mathcal{C} = \langle (g_0(x), 0), (0, (1 - v)g_1(x) + vg_2(x)) \rangle$ is a separable $\mathcal{S}$-$t$-constacyclic code. Then
		\begin{enumerate}[label=(\arabic*)]
			
			\item[(i)] \( \mathcal{C}^{\perp_E} = \langle (h_0^*(x),0), (0,(1 - v)h_1^*(x) + vh_2^*(x)) \rangle \) is a separable $\mathcal{S}$-$t^{-1}$-constacyclic code, where $t^{-1}=(\lambda^{-1},\mu^{-1})$ and 
			$  h_0(x) = \frac{x^m - \lambda}{g_0(x)}$ in $\mathbb{F}_q[x]$. 
			Moreover,
			\(|\mathcal{C}^{\perp_E}| = q^{{\deg(g_0(x)) + \deg(g_1(x)) + \deg(g_2(x))}}.\)
			
			\item[(ii)] \( \mathcal{C}^{\perp_H} = \langle (h_0^\dagger(x),0), (0,(1 - v)h_1^\dagger(x) + vh_2^\dagger(x)) \rangle \) is a separable $\mathcal{S}$-$t^{-p}$-constacyclic code, where $t^{-p}=(\lambda^{-p},\mu^{-p})$, 
			$ h_0(x) = \frac{x^m - \lambda}{g_0(x)}$ in $\mathbb{F}_{p^2}[x]$.  
			Moreover,
			\(|\mathcal{C}^{\perp_H}| = p^{2({\deg(g_0(x)) + \deg(g_1(x)) + \deg(g_2(x))})}.\)
			
		\end{enumerate}
	\end{theorem}
	
	\begin{proof}
		(i) By Proposition 4.1 in \cite{Qian2026}, take $r=1$, $\xi_0=\beta$ and $\xi_1=-\beta$, it is clear that $\mathcal{C}^{\perp_E}$ is a separable $\mathcal{S}$-$t^{-1}$-constacyclic code, where  $t^{-1}=(\lambda^{-1},\mu^{-1})$. 
		
		Given that $\mathcal{C}=C_{m}\times\mathscr{C}_{n}$, it follows that $\mathcal{C}^{\perp_{E}}=C_{m}^{\perp_{E}}\times\mathscr{C}_{n}^{\perp_{E}}$. From Lemma 2.8, $C_m = \langle g_0(x) \rangle$. Then $C_{m}^{\perp_E} = \langle h_0^*(x) \rangle$ and $|C_{m}^{\perp_E}| = q^{\deg(g_0(x))}$ by \cite{MacWilliams1977}, where $ h_0(x) = \frac{x^m - \lambda}{g_0(x)}$ in $\mathbb{F}_q[x]$. According to Lemma 3.4, $\mathscr{C}_{n}^{\perp_E} = \langle (1 - v)h_1^*(x) + vh_2^*(x) \rangle$ and $|\mathscr{C}_{n}^{\perp_E}| = q^{{\deg(g_1(x)) + \deg(g_2(x))}}$. Therefore, \[ \mathcal{C}^{\perp_E} = \langle (h_0^*(x),0), (0,(1 - v)h_1^*(x) + vh_2^*(x)) \rangle.\] 
		
		Moreover, $|\mathcal{C}^{\perp_E}| = |\mathrm{C}_m^{\perp_E}| \cdot |\mathscr{C}_n^{\perp_E}| = q^{\deg(g_0(x)) + \deg(g_1(x)) + \deg(g_2(x))}$.
		
		(ii) For the Hermitian dual code, one can obtain the desired results by a proof similar to (i); thus, we omit the details.
	\end{proof}
	
	\begin{theorem}
		Suppose $\mathcal{C} = \langle (g_0(x), 0), (0, (1 - v)g_1(x) + vg_2(x)) \rangle$ is a separable $\mathcal{S}$-$t$-constacyclic code. Then
		\begin{enumerate}[label=(\arabic*)]
			\item[(i)]
			$ \Phi_M(\mathcal{C}^{\perp_E})= \langle (h^*_0(x),0), (0,h^*_1(x)h^*_2(x)) \rangle $, where $M = \begin{pmatrix} \vartheta_1 \beta & \vartheta_1 \\ - \vartheta_2 \beta & \vartheta_2 \end{pmatrix}$, $\beta, \vartheta_1, \vartheta_2 \in \mathbb{F}_{q}\setminus \{0\};$
			\item[(ii)]
			$ \Phi_M(\mathcal{C}^{\perp_H})= \langle (h^\dagger_0(x),0), (0,h^\dagger_1(x)h^\dagger_2(x)) \rangle$, where $M = \begin{pmatrix} \vartheta_1 \beta & \vartheta_1 \\ - \vartheta_2 \beta & \vartheta_2 \end{pmatrix}$, $\beta, \vartheta_1, \vartheta_2 \in \mathbb{F}_{p^2}\setminus \{0\}$.
		\end{enumerate}
	\end{theorem}
	
	\begin{proof}
		 By Theorem 3.5 and Lemma 2.10, it is clear that \( \Phi_M(\mathcal{C}^{\perp_E})= \langle (h^*_0(x),0), (0,h^*_1(x)h^*_2(x)) \rangle \) and 
		 	$ \Phi_M(\mathcal{C}^{\perp_H})= \langle (h^\dagger_0(x),0), (0,h^\dagger_1(x)h^\dagger_2(x)) \rangle$.
	\end{proof}

	\subsection{Euclidean and Hermitian hulls of separable $\mathcal{S}$-$t$-constacyclic codes}
	
	\begin{definition}
		The Euclidean hull of an $\mathcal{S}$-linear code \( \mathcal{C} \) is characterized as
		$\Hull_{E}(\mathcal{C}) = \mathcal{C} \cap \mathcal{C}^{\perp_E}.$
		In addition,  \( \Hull_{H}(\mathcal{C}) = \mathcal{C} \cap \mathcal{C}^{\perp_H}\) is termed the Hermitian hull of an $\mathcal{S}$-linear code \( \mathcal{C} \).
	\end{definition}

	\begin{theorem}
		Suppose $\mathcal{C} = \langle (g_0(x), 0), (0, (1 - v)g_1(x) + vg_2(x)) \rangle$ is a separable $\mathcal{S}$-$t$-constacyclic code. Then
		
		\begin{enumerate}
			
			\item[(i)] $\Hull_E(\mathcal{C}) = \Hull_E(C_m)\times \Hull_E(\mathscr{C}_n);$
			
			\item[(ii)] $\Hull_E(C_m) = \langle \lcm(g_0(x),h_0^*(x))\rangle$,\\ $\Hull_E(\mathscr{C}_n) = \langle (1-v)\lcm(g_1(x),h_1^*(x))+v\lcm(g_2(x),h_2^*(x))\rangle;$
			
			\item[(iii)] $\Hull_E(\mathcal{C}) = \langle (\lcm(g_0(x), h_0^*(x)),0), (0,(1-v)\lcm(g_1(x), h_1^*(x)) + v\lcm(g_2(x), h_2^*(x))) \rangle; $
			
			\item[(iv)] $|\Hull_E(\mathcal{C}) | = q^{m+2n - \deg(\lcm(g_0(x), h_0^*(x))) - \deg(\lcm(g_1(x), h_1^*(x))) - \deg(\lcm(g_2(x), h_2^*(x)))}.$
		\end{enumerate}
	\end{theorem}
	
	\begin{proof}
		(i) Since $\mathcal{C}$ is a separable $\mathcal{S}$-$t$-constacyclic code, $\mathcal{C}=C_{m}\times\mathscr{C}_{n}$ yields $\mathcal{C}^{\perp_E}=C_{m}^{\perp_E}\times\mathscr{C}_{n}^{\perp_{E}}$. 
		For any $\mathbf{c} = (\mathbf{c}_1, \mathbf{c}_2) \in \Hull_E(\mathcal{C}) = \mathcal{C}  \cap \mathcal{C}^{\perp_E}$, we have $\mathbf{c} \in \mathcal{C} = C_{m}\times\mathscr{C}_{n}$, then $\mathbf{c}_1 \in C_{m}$, $\mathbf{c}_2 \in \mathscr{C}_{n}$. Simultaneously, $\mathbf{c} \in \mathcal{C}^{\perp_{E}} = C_{m}^{\perp_E}\times\mathscr{C}_{n}^{\perp_{E}}$, then $\mathbf{c}_1 \in C_{m}^{\perp_{E}}$, $\mathbf{c}_2 \in \mathscr{C}_{n}^{\perp_{E}}$. 
		Thus, \[\mathbf{c}_1 \in C_m \cap C_{m}^{\perp_{E}} = \Hull_E(C_m),\mathbf{c}_2 \in \mathscr{C}_{n} \cap \mathscr{C}_{n}^{\perp_{E}} = \Hull_E(\mathscr{C}_{n}),\] 
		i.e., $\mathbf{c} \in \Hull_E(C_{m}) \times \Hull_E(\mathscr{C}_{n})$. Therefore, $\Hull_E(\mathcal{C}) \subset \Hull_E(C_{m}) \times \Hull_E(\mathscr{C}_{n})$.

		On the other hand, for every $\mathbf{d} = (\mathbf{d}_1, \mathbf{d}_2) \in \Hull_E(C_m) \times \Hull_E(\mathscr{C}_{n}) $, then $\mathbf{d}_1 \in \Hull_E(C_m) = C_{m} \cap C_{m}^{\perp_{E}}$ and $\mathbf{d}_2 \in \Hull_E(\mathscr{C}_n) = \mathscr{C}_{n} \cap \mathscr{C}_{n}^{\perp_{E}}$. Hence, $\mathbf{d} = (\mathbf{d}_1,\mathbf{d}_2) \in C_m \times \mathscr{C}_{n}$ and $\mathbf{d} = (\mathbf{d}_1,\mathbf{d}_2) \in C_{m}^{\perp_{E}} \times \mathscr{C}_{n}^{\perp_{E}}$. 
		Clearly, \[\mathbf{d} = (\mathbf{d}_1, \mathbf{d}_2) \in (C_{m} \times \mathscr{C}_{n}) \cap (C_{m}^{\perp_{E}} \times \mathscr{C}_{n}^{\perp_{E}}) = \mathcal{C} \cap \mathcal{C}^{\perp_{E}} = \Hull_E(\mathcal{C}).\] 
		Consequently, $\Hull_E(C_m) \times \Hull_E(\mathscr{C}_{n}) \subset \Hull_E(\mathcal{C})$.
		
		As a result, $\Hull_E(\mathcal{C}) = \Hull_E(C_m) \times \Hull_E(\mathscr{C}_{n}) $.

		(ii) As in the proof of Theorem 1 in \cite{Sangwisut2015}, \[\Hull_E(C_{m}) = C_{m} \cap C_{m}^{\perp_{E}} = \langle \lcm(g_0(x),h_0^*(x))\rangle,\]
		and similar to Theorem 2 in \cite{Tian2024}, \[\Hull_E(\mathscr{C}_{n}) = \mathscr{C}_{n} \cap \mathscr{C}_{n}^{\perp_{E}} = \langle (1-v)\lcm(g_1(x),h_1^*(x))+v\lcm(g_2(x),h_2^*(x))\rangle.\]
		
		(iii) Due to (i) and (ii), it follows that 
		\[\Hull_E(\mathcal{C}) = \langle (\lcm(g_0(x), h_0^*(x)),0), (0,(1-v)\lcm(g_1(x), h_1^*(x)) + v\lcm(g_2(x), h_2^*(x))) \rangle. \]
		
		(iv)
		 By (i) and (ii),
		
		\begin{align*}
		|\Hull_E(\mathcal{C})| &= |\Hull_E(C_m)| \cdot |\Hull_E(\mathscr{C}_n)| \\
		&= q^{m+2n - \deg(\operatorname{lcm}(g_0(x), h_0^*(x))) - \deg(\operatorname{lcm}(g_1(x), h_1^*(x))) - \deg(\operatorname{lcm}(g_2(x), h_2^*(x)))}.
		\end{align*}
	\end{proof}
	
	\begin{theorem}
		Suppose $\mathcal{C} = \langle (g_0(x), 0), (0, (1 - v)g_1(x) + vg_2(x)) \rangle$ is a separable $\mathcal{S}$-$t$-constacyclic code. Then
		\begin{enumerate}
			
			\item[(i)] $\Hull_H(\mathcal{C}) = \Hull_H(C_m) \times \Hull_H(\mathscr{C}_n);$
			
			\item[(ii)] $\Hull_H(C_m) = \langle \lcm(g_0(x),h_0^\dagger(x))\rangle$,\\ $\Hull_H(\mathscr{C}_n) = \langle (1-v)\lcm(g_1(x),h_1^\dagger(x))+v\lcm(g_2(x),h_2^\dagger(x))\rangle;$
			
			\item[(iii)] $\Hull_H(\mathcal{C}) = \langle (\lcm(g_0(x), h_0^\dagger(x)),0), (0,(1-v)\lcm(g_1(x), h_1^\dagger(x)) + v\lcm(g_2(x), h_2^\dagger(x))) \rangle ;$
			
			\item[(iv)] $| \Hull_H(\mathcal{C}) | = p^{2({m+2n - \deg(\lcm(g_0(x), h_0^\dagger(x))) - \deg(\lcm(g_1(x), h_1^\dagger(x))) - \deg(\lcm(g_2(x), h_2^\dagger(x)))})}.$
			
		\end{enumerate}
	\end{theorem}
	
	\begin{proof}
		One can prove this by mimicking the proof of Theorem 3.8.
	\end{proof}

	\begin{lemma}
		Suppose \( \mathcal{C} \) is an $\mathcal{S}$-linear code. Then
		\begin{enumerate}
			\item[(i)] $ \Hull_{E}(\Phi_M(\mathcal{C})) = \Phi_M(\Hull_{E}(\mathcal{C}))$, where $MM^\top=diag(\theta_1, \theta_2)$, $\theta_i \in  \mathbb{F}_{q}\setminus \{0\}$, $i=1,2$.
			\item[(ii)] $ \Hull_{H}(\Phi_M(\mathcal{C})) = \Phi_M(\Hull_{H}(\mathcal{C})) $, where $MM^\dagger=diag(\theta_1, \theta_2)$, $\theta_i \in \mathbb{F}_{p^2}\setminus \{0\}$, $i=1,2$.
		\end{enumerate}
	\end{lemma}
	
	\begin{proof}
		(i) Firstly, for any $\mathbf{c} \in \Hull_E(\mathcal{C}) = \mathcal{C} \cap \mathcal{C}^{\perp_E}$, then $\Phi_M(\mathbf{c}) \in \Phi_M(\mathcal{C}) \cap \Phi_M(\mathcal{C}^{\perp_E}) $. It follows from Lemma 3.2 that \[\Phi_M(\mathbf{c}) \in \Hull_E(\Phi_M(\mathcal{C})),\] where $MM^\top=diag(\theta_1, \theta_2)$, $\theta_i \in  \mathbb{F}_{q}\setminus \{0\}$, $i=1,2$. 
		Therefore, \( \Phi_M(\Hull_{E}(\mathcal{C})) \subset \Hull_{E}(\Phi_M(\mathcal{C})) \).
		
		On the other hand, suppose \( \mathbf{d} \in \Hull_{E}(\Phi_M(\mathcal{C})) = \Phi_M(\mathcal{C}) \cap \Phi_M(\mathcal{C})^{\perp_E} \), using Lemma 3.2, one can deduce that \( \mathbf{d} \in \Phi_M(\mathcal{C}) \) and \( \mathbf{d} \in \Phi_M(\mathcal{C})^{\perp_E} = \Phi_M(\mathcal{C}^{\perp_E}) \), then there is an \( \mathbf{a} \in \mathcal{C} \) such that \( \mathbf{d} = \Phi_M(\mathbf{a}) \) and there exists a \( \mathbf{b} \in \mathcal{C}^{\perp_E} \) such that \( \mathbf{d} = \Phi_M(\mathbf{b}) \). Since \( \Phi_M \) is  bijection, \( \mathbf{a} = \mathbf{b} \), hence, \[ \mathbf{d} \in \Phi_M(\Hull_{E}(\mathcal{C})). \] It can be deduced that $\Hull_{E}(\Phi_M(\mathcal{C})) \subset \Phi_M(\Hull_{E}(\mathcal{C}))$.
		
		In conclusion, \( \Hull_{E}(\Phi_M(\mathcal{C})) = \Phi_M(\Hull_{E}(\mathcal{C})) \).
		
		(ii) The conclusion is obtained by (i) and Lemma 3.3.
	\end{proof}	
	
	Now let us proceed to further derive the generator polynomials of the Euclidean hulls and Hermitian hulls of their Gray images.
	
	\begin{theorem}
		Assume that $\mathcal{C} = \langle (g_0(x), 0), (0, (1 - v)g_1(x) + vg_2(x)) \rangle$ is a separable $\mathcal{S}$-$t$-constacyclic code and $M = \begin{pmatrix} \vartheta_1 \beta & \vartheta_1 \\ - \vartheta_2 \beta & \vartheta_2 \end{pmatrix}$, where $\beta, \vartheta_1, \vartheta_2 \in \mathbb{F}_{q}\setminus \{0\}$. Then
		\begin{enumerate}[label=(\arabic*)]
			
			\item[(i)] $\Hull_E(\Phi_M(\mathcal{C}))= \Hull_E(C_m) \times \Hull_E(\phi_M(\mathscr{C}_n));$
			
			\item[(ii)] $\Hull_E(\phi_M(\mathscr{C}_n))=\langle \operatorname{lcm}(g_{1}(x)g_{2}(x), h_{1}^*(x)h_{2}^*(x)) \rangle;$
			
			\item[(iii)]
			$\operatorname{Hull}_{E}(\Phi_M(\mathcal{C})) = \langle (\operatorname{lcm}(g_0(x), h_0^*(x)),0),(0,\lcm(g_{1}(x)g_{2}(x), h_{1}^*(x)h_{2}^*(x))) \rangle;$

			\item[(iv)]
			$|\Hull_E(\Phi_M(\mathcal{C}))| = q^{m+2n - \deg(\lcm(g_0(x), h_0^*(x))) - \deg(\lcm(g_1(x), h_1^*(x))) - \deg(\lcm(g_2(x), h_2^*(x)))},$ and
			\begin{align*}
				\dimension_{\mathbb{F}_{q}}(\Hull_E(\Phi_M(\mathcal{C}))) &=m+2n - \deg(\lcm(g_0(x), h_0^*(x)))\\ 
				&- \deg(\lcm(g_1(x), h_1^*(x))) - \deg(\lcm(g_2(x), h_2^*(x))).
			\end{align*} 
		\end{enumerate}
	\end{theorem}

	\begin{proof}
		(i) According to $\Phi_M(\mathcal{C}) = C_{m}\times \phi_M(\mathscr{C}_{n})$ in Lemma 2.9, $\Phi_M(\mathcal{C}^{\perp_E}) = C_{m}^{\perp_E} \times \phi_M(\mathscr{C}_{n}^{\perp_E})$. 
		For any $\mathbf{c} = (\mathbf{c}_1, \mathbf{c}_2) \in \Hull_E(\Phi_M(\mathcal{C})) = \Phi_M(\mathcal{C})  \cap \Phi_M(\mathcal{C})^{\perp_E}$. By Lemma 3.2, 
		\[\mathbf{c} \in \Phi_M(\mathcal{C})  \cap \Phi_M(\mathcal{C}^{\perp_E}).\] Thus, $\mathbf{c} \in \Phi_M(\mathcal{C}) = C_{m} \times \phi_M(\mathscr{C}_n)$, which means that $\mathbf{c}_1 \in C_{m}$ and $\mathbf{c}_2 \in \phi_M(\mathscr{C}_n)$. Meanwhile, $\mathbf{c} \in \Phi_M(\mathcal{C}^{\perp_E}) = C_{m}^{\perp_E} \times \phi_M(\mathscr{C}_n^{\perp_E})$, then $\mathbf{c}_1 \in C_{m}^{\perp_E}$ and $\mathbf{c}_2 \in \phi_M(\mathscr{C}_n^{\perp_E})$. Similar to Corollary 4.8 in \cite{Zhu2011}, $\phi_M(\mathscr{C}_n^{\perp_E}) = \phi_M(\mathscr{C}_n)^{\perp_E}$, so $\mathbf{c}_2 \in \phi_M(\mathscr{C}_n)^{\perp_E}$. Hence, \[\mathbf{c}_1 \in C_{m} \cap C_{m}^{\perp_E} = \Hull_E(C_m),\] 
		\[\mathbf{c}_2 \in \phi_M(\mathscr{C}_n) \cap \phi_M(\mathscr{C}_n)^{\perp_E} = \Hull_E(\phi_M(\mathscr{C}_n)),\]
		i.e., $\mathbf{c} \in \Hull_E(C_{m}) \times \Hull_E(\phi_M(\mathscr{C}_{n}))$. 
		Therefore, $\Hull_E(\Phi_M(\mathcal{C})) \subset \Hull_E(C_{m}) \times \Hull_E(\phi_M(\mathscr{C}_{n}))$.

		On the other hand, for every $\mathbf{d} = (\mathbf{d}_1, \mathbf{d}_2) \in \Hull_E(C_m) \times \Hull_E(\phi_M(\mathscr{C}_{n})) $, then \[\mathbf{d}_1 \in \Hull_E(C_m)
		= C_m \cap C_m^{\perp_{E}},\]
		\[\mathbf{d}_2 \in \Hull_E(\phi_M(\mathscr{C}_n)) = \phi_M(\mathscr{C}_n) \cap \phi_M(\mathscr{C}_n)^{\perp_{E}}.\] 
		As in the proof of Corollary 4.8 in \cite{Zhu2011}, $\phi_M(\mathscr{C}_n)^{\perp_{E}}= \phi_M(\mathscr{C}_n^{\perp_{E}})$, then $\mathbf{d}_2 \in \phi_M(\mathscr{C}_n) \cap \phi_M(\mathscr{C}_n^{\perp_{E}})$. Hence, $\mathbf{d} = (\mathbf{d}_1, \mathbf{d}_2) \in C_m \times \phi_M(\mathscr{C}_n)$ and $\mathbf{d} = (\mathbf{d}_1, \mathbf{d}_2) \in C_{m}^{\perp_{E}} \times \phi_M(\mathscr{C}_n^{\perp_{E}})$. Thus, 
		\[\mathbf{d} = (\mathbf{d}_1, \mathbf{d}_2) \in (C_{m}\times \phi_M(\mathscr{C}_{n})) \cap (C_{m}^{\perp_E} \times \phi_M(\mathscr{C}_{n}^{\perp_E})),\]
		that is, $\mathbf{d} \in \Phi_M(\mathcal{C}) \cap \Phi_M(\mathcal{C}^{\perp_E})$. According to Lemma 3.2, $\mathbf{d} \in \Phi_M(\mathcal{C}) \cap \Phi_M(\mathcal{C})^{\perp_E} = \Hull_E(\Phi_M(\mathcal{C}))$.
		Consequently, $\Hull_E(C_m) \times \Hull_E(\phi_M(\mathscr{C}_{n})) \subset \Hull_E(\Phi_M(\mathcal{C}))$.
		
		As a result, $\Hull_E(\Phi_M(\mathcal{C})) = \Hull_E(C_m) \times \Hull_E(\phi_M(\mathscr{C}_{n})) $.

		(ii) From Lemma 2.2, $\phi_M(\mathscr{C}_n) = \langle g_1(x)g_2(x) \rangle$, where $M = \begin{pmatrix} \vartheta_1 \beta & \vartheta_1 \\ - \vartheta_2 \beta & \vartheta_2 \end{pmatrix}$ and $\beta, \vartheta_1, \vartheta_2 \in \mathbb{F}_{q}\setminus \{0\}$. By Lemma 3.4, $\mathscr{C}_n^{\perp_{E}} = \langle (1 - v)h_1^*(x) + vh_2^*(x)) \rangle$. 
		According to Lemma 2.2,
		$\phi_M(\mathscr{C}_n^{\perp_{E}}) = \langle h_1^*(x)h_2^*(x) \rangle $. Following the proof of Corollary 4.8 in \cite{Zhu2011}, we have $\phi_M(\mathscr{C}_n^{\perp_{E}}) = \phi_M(\mathscr{C}_n)^{\perp_{E}}$, then $\phi_M(\mathscr{C}_{n})^{\perp_{E}}= \langle h_1^*(x)h_2^*(x) \rangle$. Accordingly, similar to Theorem 4 in \cite{Tian2024}, one can deduce that
		\[\Hull_E(\phi_M(\mathscr{C}_{n})) = \phi_M(\mathscr{C}_{n}) \cap \phi_M(\mathscr{C}_{n})^{\perp_{E}} = \langle \lcm(g_1(x)g_2(x),h_1^*(x)h_2^*(x))\rangle.\]
		
		(iii) By (i), (ii) and Theorem 3.8 (ii), 
		\[\operatorname{Hull}_{E}(\Phi_M(\mathcal{C})) = \langle (\operatorname{lcm}(g_0(x), h_0^*(x)),0),(0,\lcm(g_{1}(x)g_{2}(x), h_{1}^*(x)h_{2}^*(x))) \rangle. \]
		
		(iv)
		According to Lemma 3.10, Theorem 3.8 (iv), and the bijectivity of $\Phi_M$,
		
		\begin{align*}
			|\Hull_E(\Phi_M(\mathcal{C}))| &= |\Phi_M(\Hull_E(\mathcal{C}))| \\
			&= |\Hull_E(\mathcal{C})|\\
			&= q^{m+2n - \deg(\lcm(g_0(x), h_0^*(x))) - \deg(\lcm(g_1(x), h_1^*(x))) - \deg(\lcm(g_2(x), h_2^*(x)))}.
		\end{align*}
		and
		\begin{align*}
		\dimension_{\mathbb{F}_{q}}(\Hull_E(\Phi_M(\mathcal{C}))) &=m+2n - \deg(\lcm(g_0(x), h_0^*(x)))\\ 
		&- \deg(\lcm(g_1(x), h_1^*(x))) - \deg(\lcm(g_2(x), h_2^*(x))).
	\end{align*}
	
	\end{proof}

	\begin{theorem}
		Suppose $\mathcal{C} = \langle (g_0(x), 0), (0, (1 - v)g_1(x) + vg_2(x)) \rangle$ is a separable $\mathcal{S}$-$t$-constacyclic code and $M = \begin{pmatrix} \vartheta_1 \beta & \vartheta_1 \\ - \vartheta_2 \beta & \vartheta_2 \end{pmatrix}$, where $\beta, \vartheta_1, \vartheta_2 \in \mathbb{F}_{p^2}\setminus \{0\}$. Then
		\begin{enumerate}[label=(\arabic*)]
			
			\item[(i)] $\Hull_H(\Phi_M(\mathcal{C}))= \Hull_H(C_m) \times \Hull_H(\phi_M(\mathscr{C}_n));$
			
			\item[(ii)] $\Hull_H(\phi_M(\mathscr{C}_n))=\langle \operatorname{lcm}(g_{1}(x)g_{2}(x), h_{1}^\dagger(x)h_{2}^\dagger(x)) \rangle;$
			
			\item[(iii)]
			$\operatorname{Hull}_{H}(\Phi_M(\mathcal{C})) = \langle (\operatorname{lcm}(g_0(x), h_0^\dagger(x)),0),(0,\lcm(g_{1}(x)g_{2}(x), h_{1}^\dagger(x)h_{2}^\dagger(x))) \rangle;$

			\item[(iv)]
			$|\Hull_H(\Phi_M(\mathcal{C}))| = p^{2(m+2n - \deg(\lcm(g_0(x), h_0^\dagger(x))) - \deg(\lcm(g_1(x), h_1^\dagger(x))) - \deg(\lcm(g_2(x), h_2^\dagger(x))))},$ and
			\begin{align*}
			\dimension_{\mathbb{F}_{p^2}}(\Hull_H(\Phi_M(\mathcal{C}))) &= m+2n - \deg(\lcm(g_0(x), h_0^\dagger(x)))\\ 
			&- \deg(\lcm(g_1(x), h_1^\dagger(x))) - \deg(\lcm(g_2(x), h_2^\dagger(x))).
			\end{align*}
		\end{enumerate}
	\end{theorem}
	
	\begin{proof}
		The conclusions can be similarly  derived as Theorem 3.11.
	\end{proof}

	\subsection{Euclidean and Hermitian sums of separable $\mathcal{S}$-$t$-constacyclic codes}
	
	Next, the Euclidean sums and Hermitian sums of separable $\mathcal{S}$-$t$-constacyclic codes are introduced.
	
	\begin{definition}
		The Euclidean sum of an \( \mathcal{S} \)-linear code \( \mathcal{C} \) is defined as
		$\Sum_{E}(\mathcal{C}) = \mathcal{C} + \mathcal{C}^{\perp_E}= \{\mathbf{c}_1 + \mathbf{c}_2 \mid \mathbf{c}_1 \in \mathcal{C}, \mathbf{c}_2 \in \mathcal{C}^{\perp_E}\}.$
		In addition,  \( \Sum_{H}(\mathcal{C}) = \mathcal{C} + \mathcal{C}^{\perp_H}= \{\mathbf{c}_1 + \mathbf{c}_2 \mid \mathbf{c}_1 \in \mathcal{C}, \mathbf{c}_2 \in \mathcal{C}^{\perp_H}\}\) is called the Hermitian sum of an \( \mathcal{S} \)-linear code \( \mathcal{C} \).
	\end{definition}

	\begin{theorem}
		Suppose $\mathcal{C} = \langle (g_0(x), 0), (0, (1 - v)g_1(x) + vg_2(x)) \rangle$ is a separable $\mathcal{S}$-$t$-constacyclic code. Then
		
		\begin{enumerate}[label=(\arabic*)]
			
			\item[(i)] $\Sum_E(\mathcal{C}) = \Sum_E(C_m)\times \Sum_E(\mathscr{C}_n);$ 
			
			\item[(ii)] $\Sum_E(C_m) = \langle \gcd(g_0(x),h_0^*(x))\rangle$,\\ $\Sum_E(\mathscr{C}_n) = \langle (1-v)\gcd(g_1(x),h_1^*(x))+v\gcd(g_2(x),h_2^*(x))\rangle;$
			
			\item[(iii)] $\Sum_E(\mathcal{C}) = \langle (\gcd(g_0(x), h_0^*(x)),0), (0,(1-v)\gcd(g_1(x), h_1^*(x)) + v\gcd(g_2(x), h_2^*(x))) \rangle; $
			
			\item[(iv)] $|\Sum_E(\mathcal{C}) | = q^{m+2n - \deg(\gcd(g_0(x), h_0^*(x))) - \deg(\gcd(g_1(x), h_1^*(x))) - \deg(\gcd(g_2(x), h_2^*(x)))}.$
			
		\end{enumerate}
	\end{theorem}
	
	\begin{proof}
		(i) First, for any $\mathbf{c} \in \Sum_E(\mathcal{C}) = \mathcal{C} + \mathcal{C}^{\perp_E}$, suppose $\mathbf{c} = \mathbf{c}_1 + \mathbf{c}_2$, then $\mathbf{c}_1 \in \mathcal{C}$ and $\mathbf{c}_2 \in \mathcal{C}^{\perp_E}$.
		Next, since $\mathcal{C}$ is a separable $\mathcal{S}$-$t$-constacyclic code, $\mathcal{C}=C_{m}\times\mathscr{C}_{n}$ yields $\mathcal{C}^{\perp_E}=C_{m}^{\perp_E}\times\mathscr{C}_{n}^{\perp_{E}}$. Then
		assume that $\mathbf{c}_1 = (\mathbf{c}_{11},\mathbf{c}_{12})$ and $\mathbf{c}_2 = (\mathbf{c}_{21},\mathbf{c}_{22})$, where $\mathbf{c}_{11} \in C_{m}$, $\mathbf{c}_{12} \in \mathscr{C}_{n}$, $\mathbf{c}_{21} \in C_{m}^{\perp_{E}}$, $\mathbf{c}_{22} \in \mathscr{C}_{n}^{\perp_{E}}$. Thus, \[\mathbf{c} = (\mathbf{c}_{11} + \mathbf{c}_{21}, \mathbf{c}_{12} + \mathbf{c}_{22}),\] 
		where $\mathbf{c}_{11} + \mathbf{c}_{21} \in C_{m} + C_{m}^{\perp_{E}}$, i.e., $\mathbf{c}_{11} + \mathbf{c}_{21} \in \Sum_E(C_{m})$. Similarly, $\mathbf{c}_{12} + \mathbf{c}_{22} \in \Sum_E(\mathscr{C}_{n})$. Hence, $\mathbf{c} \in \Sum_E(C_{m}) \times \Sum_E(\mathscr{C}_{n})$. 
		Therefore, $\Sum_E(\mathcal{C}) \subset \Sum_E(C_{m}) \times \Sum_E(\mathscr{C}_{n})$.

		On the other hand, for any $\mathbf{d} \in \Sum_E(C_m) \times \Sum_E(\mathscr{C}_{n}) $, let $\mathbf{d} = (\mathbf{d}_1, \mathbf{d}_2)$, then $\mathbf{d}_1 \in \Sum_E(C_m)$ and $\mathbf{d}_2 \in \Sum_E(\mathscr{C}_n)$. Assume further that $\mathbf{d}_1 = \mathbf{d}_{11} + \mathbf{d}_{12}$ and $\mathbf{d}_2 = \mathbf{d}_{21} + \mathbf{d}_{22}$, where $\mathbf{d}_{11} \in C_m$, $\mathbf{d}_{12} \in C_m^{\perp_{E}}$, $\mathbf{d}_{21} \in \mathscr{C}_n$ and $\mathbf{d}_{22} \in \mathscr{C}_n^{\perp_{E}}$. Clearly, \[\mathbf{d} = (\mathbf{d}_{11} + \mathbf{d}_{12}, \mathbf{d}_{21} + \mathbf{d}_{22}) = (\mathbf{d}_{11}, \mathbf{d}_{21}) + (\mathbf{d}_{12}, \mathbf{d}_{22}),\] 
		where $(\mathbf{d}_{11}, \mathbf{d}_{21}) \in C_{m}\times\mathscr{C}_{n}$, $(\mathbf{d}_{12}, \mathbf{d}_{22}) \in C_{m}^{\perp_E}\times\mathscr{C}_{n}^{\perp_{E}}$. Hence, $\mathbf{d} \in \mathcal{C} + \mathcal{C}^{\perp_{E}}$, that is, $\mathbf{d} \in \Sum_E(\mathcal{C})$. Consequently, $\Sum_E(C_m) \times \Sum_E(\mathscr{C}_{n}) \subset \Sum_E(\mathcal{C})$.
		
		In summary, $\Sum_E(\mathcal{C}) = \Sum_E(C_m) \times \Sum_E(\mathscr{C}_{n}) $.

		(ii) Similar to Theorem 3.2 in \cite{Dougherty2024}, \[\Sum_E(C_{m}) = C_{m} + C_{m}^{\perp_{E}} = \Hull_E(C_{m})^{\perp_{E}} = \langle \gcd(g_0(x),h_0^*(x))\rangle.\]
		And as in the proof of Theorem 2 in \cite{Tian2024}, \[\Sum_E(\mathscr{C}_{n}) = \mathscr{C}_{n} + \mathscr{C}_{n}^{\perp_{E}} = \Hull_E(\mathscr{C}_{n})^{\perp_{E}} = \langle (1-v)\gcd(g_1(x),h_1^*(x))+v\gcd(g_2(x),h_2^*(x))\rangle.\]
		
		(iii) Due to  (i) and (ii), it can be deduced that
		\[\Sum_E(\mathcal{C}) = \langle (\gcd(g_0(x), h_0^*(x)),0), (0,(1-v)\gcd(g_1(x), h_1^*(x)) + v\gcd(g_2(x), h_2^*(x))) \rangle. \]
		
		(iv) From (i), (ii), it is evident that
		
		\begin{align*}
		|\Sum_E(\mathcal{C})| &= |\Sum_E(C_m)| \cdot |\Sum_E(\mathscr{C}_n)|\\
		&= q^{m+2n - \deg(\gcd(g_0(x), h_0^*(x))) - \deg(\gcd(g_1(x), h_1^*(x))) - \deg(\gcd(g_2(x), h_2^*(x)))}.
		\end{align*}
	\end{proof}
	
	\begin{theorem}
		Assume that $\mathcal{C} = \langle (g_0(x), 0), (0, (1 - v)g_1(x) + vg_2(x)) \rangle$ is a separable $\mathcal{S}$-$t$-constacyclic code. Then
		\begin{enumerate}
			
			\item[(i)] $\Sum_H(\mathcal{C}) = \Sum_H(C_m) \times \Sum_H(\mathscr{C}_n);$
			
			\item[(ii)] $\Sum_H(C_m) = \langle \gcd(g_0(x),h_0^\dagger(x))\rangle$,\\ $\Sum_H(\mathscr{C}_n) = \langle (1-v)\gcd(g_1(x),h_1^\dagger(x))+v\gcd(g_2(x),h_2^\dagger(x))\rangle;$
			
			\item[(iii)] $\Sum_H(\mathcal{C}) = \langle (\gcd(g_0(x), h_0^\dagger(x)),0), (0,(1-v)\gcd(g_1(x), h_1^\dagger(x)) + v\gcd(g_2(x), h_2^\dagger(x))) \rangle ;$
			
			\item[(iv)] $| \Sum_H(\mathcal{C}) | = p^{2({m+2n - \deg(\gcd(g_0(x), h_0^\dagger(x))) - \deg(\gcd(g_1(x), h_1^\dagger(x))) - \deg(\gcd(g_2(x), h_2^\dagger(x))})}.$
			
		\end{enumerate}
	\end{theorem}
	
	\begin{proof}
		The required results can be derived analogously to Theorem 3.14.
	\end{proof}
	
	Below, the Euclidean and Hermitian sums of Gray imges of separable $\mathcal{S}$-$t$-constacyclic codes are discussed.

	\begin{lemma}
		Suppose \( \mathcal{C} \) is an \( \mathcal{S} \)-linear code. Then
		\begin{enumerate}
			\item[(i)] \( \Sum_{E}(\Phi_M(\mathcal{C})) = \Phi_M(\Sum_{E}(\mathcal{C})) \), where $MM^\top=diag(\theta_1, \theta_2)$, $\theta_i \in  \mathbb{F}_{q}\setminus \{0\}$, $i=1,2;$
			\item[(ii)] \( \Sum_{H}(\Phi_M(\mathcal{C})) = \Phi_M(\Sum_{H}(\mathcal{C})) \), where $MM^\dagger=diag(\theta_1, \theta_2)$, $\theta_i \in \mathbb{F}_{p^2}\setminus \{0\}$, $i=1,2$.
		\end{enumerate}
	\end{lemma}
	
	\begin{proof}
		(i) For any $\mathbf{c} \in \Sum_E(\mathcal{C}) = \mathcal{C} + \mathcal{C}^{\perp_E}$, then $\Phi_M(\mathbf{c}) \in \Phi_M(\mathcal{C}) + \Phi_M(\mathcal{C}^{\perp_E})$, according to Lemma 3.2, $\Phi_M(\mathbf{c}) \in \Sum_E(\Phi_M(\mathcal{C}))$, where $MM^\top=diag(\theta_1, \theta_2)$, $\theta_i \in  \mathbb{F}_{q}\setminus \{0\}$, $i=1,2$. Thus, $\Phi_M(\Sum_E(\mathcal{C})) \subset \Sum_E(\Phi_M(\mathcal{C}))$.
		
		On the other hand, let $\mathbf{c} \in \Sum_{E}(\Phi_M(\mathcal{C})) = \Phi_M(\mathcal{C}) + \Phi_M(\mathcal{C})^{\perp_E}$, then $\mathbf{c} \in \Phi_M(\mathcal{C}) + \Phi_M(\mathcal{C}^{\perp_E})$ by Lemma 3.2. Therefore $\mathbf{c} = \Phi_M(\mathbf{x}) + \Phi_M(\mathbf{y})$ for some $\mathbf{x} \in \mathcal{C}$, $\mathbf{y} \in \mathcal{C}^{\perp_E}$. Since $\Phi_M$ is linear, $\mathbf{c} = \Phi_M(\mathbf{x} + \mathbf{y}) \in \Phi_M(\mathcal{C} + \mathcal{C}^{\perp_E}) = \Phi_M(\Sum_{E}(\mathcal{C}))$. Hence, $\Sum_{E}(\Phi_M(\mathcal{C})) \subset \Phi_M(\Sum_E(\mathcal{C}))$.
		
		Therefore, \(\Phi_M(\Sum_E(\mathcal{C})) = \Sum_E(\Phi_M(\mathcal{C}))\).
		
		(ii) From (i), the conclusion follows similarly.
	\end{proof}

	\begin{theorem}
		Suppose $\mathcal{C} = \langle (g_0(x), 0), (0, (1 - v)g_1(x) + vg_2(x)) \rangle$ is a separable $\mathcal{S}$-$t$-constacyclic code and $M = \begin{pmatrix} \vartheta_1 \beta & \vartheta_1 \\ - \vartheta_2 \beta & \vartheta_2 \end{pmatrix}$, where $\beta, \vartheta_1, \vartheta_2 \in \mathbb{F}_{q}\setminus \{0\}$. Then,
		\begin{enumerate}[label=(\arabic*)]
			
			\item[(i)] $\Sum_E(\Phi_M(\mathcal{C}))= \Sum_E(C_m) \times \Sum_E(\phi_M(\mathscr{C}_n)).$ Moreover, 
			\[d(\Sum_E(\Phi_M(\mathcal{C}))) = \min\{d(\Sum_E(C_m)), d(\Sum_E(\phi_M(\mathscr{C}_n)))\}; \]
			
			\item[(ii)] $\Sum_E(\phi_M(\mathscr{C}_n))=\langle \operatorname{gcd}(g_{1}(x)g_{2}(x), h_{1}^*(x)h_{2}^*(x)) \rangle;$
			
			\item[(iii)]
			$\operatorname{Sum}_{E}(\Phi_M(\mathcal{C})) = \langle (\operatorname{gcd}(g_0(x), h_0^*(x)),0),(0,\gcd(g_{1}(x)g_{2}(x), h_{1}^*(x)h_{2}^*(x))) \rangle;$

			\item[(iv)]
			$|\Sum_E(\Phi_M(\mathcal{C}))| = q^{m+2n - \deg(\gcd(g_0(x), h_0^*(x))) - \deg(\gcd(g_1(x), h_1^*(x))) - \deg(\gcd(g_2(x), h_2^*(x)))}.$

		\end{enumerate}
	\end{theorem}
	
	\begin{proof}
		
		(i) For any $\mathbf{c} \in \Sum_E(\Phi_M(\mathcal{C})) = \Phi_M(\mathcal{C})  + \Phi_M(\mathcal{C})^{\perp_E}$, by Lemma 3.2, \[\mathbf{c} \in \Phi_M(\mathcal{C})  + \Phi_M(\mathcal{C}^{\perp_E}).\] 
		Let $\mathbf{c} = \mathbf{c}_1 + \mathbf{c}_2$, then $\mathbf{c}_1 \in \Phi_M(\mathcal{C})$ and $\mathbf{c}_2 \in \Phi_M(\mathcal{C}^{\perp_E})$.
		From $\Phi_M(\mathcal{C}) = C_{m}\times \phi_M(\mathscr{C}_{n})$ in Lemma 2.9, $\Phi_M(\mathcal{C}^{\perp_E}) = C_{m}^{\perp_E} \times \phi_M(\mathscr{C}_{n}^{\perp_E})$.
		So, suppose $\mathbf{c}_1 = (\mathbf{c}_{11},\mathbf{c}_{12})$ and $\mathbf{c}_2 = (\mathbf{c}_{21},\mathbf{c}_{22})$, where $\mathbf{c}_{11} \in C_{m}$, $\mathbf{c}_{12} \in \phi_M(\mathscr{C}_{n})$, $\mathbf{c}_{21} \in C_{m}^{\perp_{E}}$ and $\mathbf{c}_{22} \in \phi_M(\mathscr{C}_{n}^{\perp_{E}})$. Thus, \[\mathbf{c} = (\mathbf{c}_{11} + \mathbf{c}_{21}, \mathbf{c}_{12} + \mathbf{c}_{22}),\] 
		where $\mathbf{c}_{11} + \mathbf{c}_{21} \in C_{m} + C_{m}^{\perp_{E}}$, i.e., $\mathbf{c}_{11} + \mathbf{c}_{21} \in \Sum_E(C_{m})$. Likewise, similar to Corollary 4.8 in \cite{Zhu2011}, 
		\[\mathbf{c}_{12} + \mathbf{c}_{22} \in \phi_M(\mathscr{C}_{n}) + \phi_M(\mathscr{C}_{n}^{\perp_{E}}) = \phi_M(\mathscr{C}_{n}) + \phi_M(\mathscr{C}_{n})^{\perp_{E}} = \Sum_E(\phi_M(\mathscr{C}_{n})).\] 
		Thus, $\mathbf{c} \in \Sum_E(C_{m}) \times \Sum_E(\phi_M(\mathscr{C}_{n}))$. 
		Therefore, $\Sum_E(\Phi_M(\mathcal{C})) \subset \Sum_E(C_{m}) \times \Sum_E(\phi_M(\mathscr{C}_{n}))$.

		On the other hand, for every $\mathbf{d} \in \Sum_E(C_m) \times \Sum_E(\phi_M(\mathscr{C}_{n})) $, let $\mathbf{d} = (\mathbf{d}_1, \mathbf{d}_2)$, then 
		\[\mathbf{d}_1 \in \Sum_E(C_m)
		= C_m + C_m^{\perp_{E}},\mathbf{d}_2 \in \Sum_E(\phi_M(\mathscr{C}_n)) = \phi_M(\mathscr{C}_n) + \phi_M(\mathscr{C}_n)^{\perp_{E}}.\] 
		As in the proof of Corollary 4.8 in \cite{Zhu2011}, $\mathbf{d}_2 \in \phi_M(\mathscr{C}_n) + \phi_M(\mathscr{C}_n^{\perp_{E}})$. Assume further that $\mathbf{d}_1 = \mathbf{d}_{11} + \mathbf{d}_{12}$ and $\mathbf{d}_2 = \mathbf{d}_{21} + \mathbf{d}_{22}$, where $\mathbf{d}_{11} \in C_m$, $\mathbf{d}_{12} \in C_m^{\perp_{E}}$, $\mathbf{d}_{21} \in \phi_M(\mathscr{C}_n)$ and $\mathbf{d}_{22} \in \phi_M(\mathscr{C}_n^{\perp_{E}})$. Then 
		\[\mathbf{d} = (\mathbf{d}_{11} + \mathbf{d}_{12}, \mathbf{d}_{21} + \mathbf{d}_{22}) = (\mathbf{d}_{11}, \mathbf{d}_{21}) + (\mathbf{d}_{12}, \mathbf{d}_{22}),\] 
		where $(\mathbf{d}_{11}, \mathbf{d}_{21}) \in C_{m} \times \phi_M(\mathscr{C}_{n}) = \Phi_M(\mathcal{C})$, $(\mathbf{d}_{12}, \mathbf{d}_{22}) \in C_{m}^{\perp_E} \times \phi_M(\mathscr{C}_{n}^{\perp_{E}}) = \Phi_M(\mathcal{C}^{\perp_{E}})$. By Lemma 3.2, $ (\mathbf{d}_{12}, \mathbf{d}_{22}) \in \Phi_M(\mathcal{C})^{\perp_{E}}$. Hence, $\mathbf{d} \in \Phi_M(\mathcal{C}) + \Phi_M(\mathcal{C})^{\perp_{E}}$, that is, $\mathbf{d} \in \Sum_E(\Phi_M(\mathcal{C}))$. 
		
		Consequently, $\Sum_E(C_m) \times \Sum_E(\phi_M(\mathscr{C}_{n})) \subset \Sum_E(\Phi_M(\mathcal{C}))$.
		
		As a result, $\Sum_E(\Phi_M(\mathcal{C})) = \Sum_E(C_m) \times \Sum_E(\phi_M(\mathscr{C}_{n})) $.
		Clearly, \[d(\Sum_E(\Phi_M(\mathcal{C}))) = \min\{d(\Sum_E(C_m)), d(\Sum_E(\phi_M(\mathscr{C}_n)))\}. \]

		(ii) From Lemma 2.2, $\phi_M(\mathscr{C}_n) = \langle g_1(x)g_2(x) \rangle $, where $M = \begin{pmatrix} \vartheta_1 \beta & \vartheta_1 \\ - \vartheta_2 \beta & \vartheta_2 \end{pmatrix}$ and $\beta, \vartheta_1, \vartheta_2 \in \mathbb{F}_{q}\setminus \{0\}$. By Lemma 3.4, $\mathscr{C}_n^{\perp_{E}} = \langle (1 - v)h_1^*(x) + vh_2^*(x)) \rangle$. Thus by Lemma 2.2, $\phi_M(\mathscr{C}_n^{\perp_{E}}) = \langle h_1^*(x)h_2^*(x)\rangle $. As in the proof of  Corollary 4.8 in \cite{Zhu2011}, $\phi_M(\mathscr{C}_n^{\perp_{E}}) = \phi_M(\mathscr{C}_n)^{\perp_{E}}$, then $\phi_M(\mathscr{C}_{n})^{\perp_{E}} = \langle h_1^*(x)h_2^*(x) \rangle$. Therefore, following the proof of Theorem 4 in \cite{Tian2024}, one can deduce that 
		\[\Sum_E(\phi_M(\mathscr{C}_{n})) = \phi_M(\mathscr{C}_{n}) + \phi_M(\mathscr{C}_{n})^{\perp_{E}} = \langle \gcd(g_1(x)g_2(x),h_1^*(x)h_2^*(x))\rangle.\]
		
		(iii) From (i), (ii) and Theorem 3.14 (ii), one can deduce that 
		\[\operatorname{Sum}_{E}(\Phi_M(\mathcal{C})) = \langle (\operatorname{gcd}(g_0(x), h_0^*(x)),0),(0,\gcd(g_{1}(x)g_{2}(x), h_{1}^*(x)h_{2}^*(x))) \rangle. \]
		
		(iv) By Lemma 3.16, Theorem 3.14 (iv), and the bijectivity of $\Phi_M$,
		
		\begin{align*}
			|\Sum_E(\Phi_M(\mathcal{C}))| &= |\Phi_M(\Sum_E(\mathcal{C}))|\\
			&= |\Sum_E(\mathcal{C})|\\
			&= q^{m+2n - \deg(\gcd(g_0(x), h_0^*(x))) - \deg(\gcd(g_1(x), h_1^*(x))) - \deg(\gcd(g_2(x), h_2^*(x)))}.
		\end{align*}
	\end{proof}

	\begin{theorem}
		Suppose $\mathcal{C} = \langle (g_0(x), 0), (0, (1 - v)g_1(x) + vg_2(x)) \rangle$ is a separable $\mathcal{S}$-$t$-constacyclic code and $M = \begin{pmatrix} \vartheta_1 \beta & \vartheta_1 \\ - \vartheta_2 \beta & \vartheta_2 \end{pmatrix}$, where $\beta, \vartheta_1, \vartheta_2 \in \mathbb{F}_{p^2}\setminus \{0\}$. Then
		\begin{enumerate}[label=(\arabic*)]
			
			\item[(i)] $\Sum_H(\Phi_M(\mathcal{C}))= \Sum_H(C_m) \times \Sum_H(\phi_M(\mathscr{C}_n)).$ Moreover, 
			\[d(\Sum_H(\Phi_M(\mathcal{C}))) = \min\{d(\Sum_H(C_m)), d(\Sum_H(\phi_M(\mathscr{C}_n)))\}; \]
			
			\item[(ii)] $\Sum_H(\phi_M(\mathscr{C}_n))=\langle \operatorname{gcd}(g_{1}(x)g_{2}(x), h_{1}^\dagger(x)h_{2}^\dagger(x)) \rangle;$
			
			\item[(iii)]
			$\operatorname{Sum}_{H}(\Phi_M(\mathcal{C})) = \langle (\operatorname{gcd}(g_0(x), h_0^\dagger(x)),0),(0,\gcd(g_{1}(x)g_{2}(x), h_{1}^\dagger(x)h_{2}^\dagger(x))) \rangle;$

			\item[(iv)]
			$|\Sum_H(\Phi_M(\mathcal{C}))| = p^{2(m+2n - \deg(\gcd(g_0(x), h_0^\dagger(x)))  - \deg(\gcd(g_1(x), h_1^\dagger(x))) - \deg(\gcd(g_2(x), h_2^\dagger(x)))}.$
			
		\end{enumerate}
	\end{theorem}
	
	\begin{proof}
		Similar to Theorem 3.17, the conclusions are clearly valid. 	
	\end{proof}

	\section{QECCs from separable $\mathcal{S}$-$t$-constacyclic codes}
	
	Next, the constructions of QECCs from the hulls and sums of separable $\mathcal{S}$-$t$-constacyclic codes will be studied, where $t=(\lambda,\mu)$,  $\mu = \beta(1-2v)$ and $\lambda, \beta \in \mathbb{F}_{q}\setminus \{0\}$.
	
	\subsection{QECCs from Euclidean hulls and sums of separable $\mathcal{S}$-$t$-constacyclic codes}
	
	First, the quantum construction X of the Euclidean dual can be introduced.
	
	\begin{lemma}(\cite{Zhang2022})
		Suppose \( C \) is an \([n, k]\) linear code over \(\mathbb{F}_q\), and denotes \( e \) as
		$e = n - k - \dimension_{\mathbb{F}_q}(\Hull_E(C)).$
		Then, there exists an \([[n + e, 2k - n + e, d]]_q\) QECC, where
		\[d \geq \min\{d(C), d(\Sum_E (C)) + 1\}.\]
	\end{lemma}
	
	The following construction follows from Lemma 4.1.
	
	\begin{theorem}
		Given that \( \mathcal{C} \) is a separable $\mathcal{S}$-$t$-constacyclic code of length $(m, n)$. Assume that \( \Phi_M(\mathcal{C}) \) is an \([m+2n, k, d(\Phi_M(\mathcal{C}))]\) constacyclic code over \( \mathbb{F}_q \) and $M = \begin{pmatrix} \vartheta_1 \beta & \vartheta_1 \\ - \vartheta_2 \beta & \vartheta_2 \end{pmatrix}$, where $\beta, \vartheta_1, \vartheta_2 \in \mathbb{F}_{q}\setminus \{0\}$. Suppose that
		$e = (m+2n) - k - \dimension_{\mathbb{F}_q}(\Hull_E(\Phi_M(\mathcal{C})))$.
		Then there exists an \([[(m + 2n) + e, 2k - (m + 2n) + e, d]]_q\) QECC, where
		\[
		d \geq \min \{d(\Phi_M(\mathcal{C})), d(\Sum_E(\Phi_M(\mathcal{C}))) + 1\}.
		\]
	\end{theorem}

	\begin{example}
		Assume that \( m = 42 \), \( n = 80 \), \( \lambda = \beta = 1 \), \( \mathcal{S} = \mathbb{F}_{7} \times (\mathbb{F}_{7} + v\mathbb{F}_{7}) \)\( (v^2 = v) \) and \( \omega \) is a generator of \( \mathbb{F}_{7}^* \). 
		Take
		\[ g_0(x) = (x + 1)^2(x + 3), g_1(x) = x + 6, g_2(x) = x^4 + x^3 + 2x^2 + 6x + 6 .\] By Lemma 2.10 and MAGMA, it can be concluded that $d(C_{m}) = 3$, $d(\phi_M(\mathscr{C}_{n})) = 3$. Thus, from Lemma 2.9, $d(\Phi_M(\mathcal{C}))=3$, and by Lemma 2.10, $k=m + 2n - \deg(g_0(x)) - \deg(g_1(x)) - \deg(g_2(x)) = 194$.
		
		Then by Theorem 3.14 (ii) and MAGMA, $d(\Sum_E(C_{m}))=3$. Likewise by Theorem 3.17 (ii) and MAGMA, $d(\Sum_E(\phi_M(\mathscr{C}_{n})))=2$. Hence, based on Theorem 3.17 (i), \[d(\Sum_E(\Phi_M(\mathcal{C}))) = \min \{d(\Sum_E(C_{m})), d(\Sum_E(\phi_M(\mathscr{C}_{n})))\} = 2.\]
		
		In addition, we use the MAGMA to obtain \[\deg(\lcm(g_0(x),h_0^*(x)))= 39, \deg(\lcm(g_1(x),h_1^*(x)))= 80, \deg(\lcm(g_2(x),h_2^*(x)))= 76.\] 
		Therefore, from Theorem 3.11 (iv) it follows that \[\dimension_{\mathbb{F}_{7}}(\Hull_E(\Phi_M(\mathcal{C})))= m + 2n - 39 - 80 - 76 = 7.\] Thus, \[e=m+2n-k-\dimension_{\mathbb{F}_{7}}(\Hull_E(\Phi_M(\mathcal{C})))=1.\]
		Then by Theorem 4.2, a QECC with parameters $[[203, 187, 3]]_{7}$ is obtained. The dimension of our code is superior to that of the QECC $[[203, 183, 3]]_{7}$ from \cite{Qian2026}.

	\end{example}

	\begin{example}
		Assume that \( m = 19 \), \( n = 22 \), \( \lambda = \beta = 1 \), \( \mathcal{S} = \mathbb{F}_{11} \times (\mathbb{F}_{11} + v\mathbb{F}_{11}) \)\( (v^2 = v) \) and \(\mathbb{F}_{11}^* = \langle \omega \rangle \). 
		Take
		\[ g_0(x) = (x + 10)(x^3 + x^2 + 2x + 10), g_1(x) = (x + 1)^3, g_2(x) = x^2 + 1 .\] By Lemma 2.10 and MAGMA, it can be concluded that $d(C_{m}) = 4$, $d(\phi_M(\mathscr{C}_{n})) = 4$. Hence, from Lemma 2.9, $d(\Phi_M(\mathcal{C}))=4$, and by Lemma 2.10, $k=m + 2n - \deg(g_0(x)) - \deg(g_1(x)) - \deg(g_2(x)) = 54$.
		
		Then by Theorem 3.14 (ii) and MAGMA, $d(\Sum_E(C_{m}))=3$. Likewise by Theorem 3.17 (ii) and MAGMA, $d(\Sum_E(\phi_M(\mathscr{C}_{n})))=4$. Thus, based on Theorem 3.17 (i), \[d(\Sum_E(\Phi_M(\mathcal{C}))) = \min \{d(\Sum_E(C_{m})), d(\Sum_E(\phi_M(\mathscr{C}_{n})))\} = 3.\]
		
		In addition, we use the MAGMA to obtain \[\deg(\lcm(g_0(x),h_0^*(x)))= 16, \deg(\lcm(g_1(x),h_1^*(x)))= 19, \deg(\lcm(g_2(x),h_2^*(x)))= 20.\] 
		Therefore, from Theorem 3.11 (iv) it follows that \[\dimension_{\mathbb{F}_{11}}(\Hull_E(\Phi_M(\mathcal{C})))= m + 2n - 16 - 19 - 20 = 8.\] Then \[e=m+2n-k-\dimension_{\mathbb{F}_{11}}(\Hull_E(\Phi_M(\mathcal{C})))=1.\]
		Hence, by Theorem 4.2, a
		QECC with parameters $[[64, 46, 4]]_{11}$ is obtained. And a higher code rate and a larger minimum distance are achieved by our code compared with the QECC $[[63, 45, 3]]_{11}$ from \cite{Caliskan2023-1}.

	\end{example}

	\subsection{QECCs from Hermitian hulls of separable $\mathcal{S}$-$t$-constacyclic codes}
	
	Next, recall that the quantum construction X of the Hermitian dual.
	
	\begin{lemma}(\cite{Lison2014})
		Suppose \( C \) is an \([n, k]\) linear code over \(\mathbb{F}_{p^2}\) and denotes \( e \) as
		$e = n - k - \dimension_{\mathbb{F}_{p^2}}(\Hull_H (C)).$
		Then, there exists an \([[n + e, 2k - n + e, d]]_p\) QECC, where
		\[d \geq \min\{d(C),d(\Sum_H (C)) + 1\}.\]
	\end{lemma}
	
	The following construction follows from Lemma 4.5.
	
	\begin{theorem}
		Assume that \( \mathcal{C} \) is a separable $\mathcal{S}$-$t$-constacyclic code of length $(m, n)$. Suppose \( \Phi_M(\mathcal{C}) \) is an \([m+2n, k, d(\Phi_M(\mathcal{C}))]\) constacyclic code over \( \mathbb{F}_{p^2} \) and $M = \begin{pmatrix} \vartheta_1 \beta & \vartheta_1 \\ - \vartheta_2 \beta & \vartheta_2 \end{pmatrix}$, where $\beta, \vartheta_1, \vartheta_2 \in \mathbb{F}_{p^2}\setminus \{0\}$. Suppose that
		$e = (m+2n) - k - \dimension_{\mathbb{F}_{p^2}}(\Hull_H(\Phi_M(\mathcal{C})) )$.
		Then there exists an \([[(m + 2n) + e, 2k - (m + 2n) + e, d]]_p\) QECC, where
		\[
		d \geq \min \{d(\Phi_M(\mathcal{C})), d(\Sum_H(\Phi_M(\mathcal{C}))) + 1\}.
		\]
	\end{theorem}
	
	\begin{example}
		Assume that \( m = 24 \), \( n = 60 \), \( \lambda = \beta = 1 \), \( \mathcal{S} = \mathbb{F}_{5^2} \times (\mathbb{F}_{5^2} + v\mathbb{F}_{5^2})(v^2 = v) \) and \(\mathbb{F}_{5^2}^* = \langle \omega \rangle \). 
		Take \[ g_0(x) = (x + 1)(x + \omega)(x + \omega^2), g_1(x) = (x + 1)(x + \omega^{2}), g_2(x) = (x + \omega)^3.\] 
		By Lemma 2.10 and MAGMA, it can be concluded that $d(C_{m})=4$ and $d(\phi_M(\mathscr{C}_{n})) = 4$.
		According to Lemma 2.9, $d(\Phi_M(\mathcal{C}))=4$, and by Lemma 2.10, $k=m+2n-\deg(g_0(x))-\deg(g_1(x))-\deg(g_2(x))=136$.
		
		Then by Theorem 3.15 (ii) and MAGMA, $d(\Sum_H(C_{m}))=3$. Likewise by Theorem 3.18 (ii) and MAGMA, $d(\Sum_H(\phi_M(\mathscr{C}_{n})))=4$. Furthermore, based on Theorem 3.18 (i), \[d(\Sum_H(\Phi_M(\mathcal{C}))) = \min \{d(\Sum_H(C_{m})), d(\Sum_H(\phi_M(\mathscr{C}_{n})))\} =3.\] 
		
		In addition, we use the MAGMA to obtain \[\deg(\lcm(g_0(x),h_0^\dagger(x)))= 22, \deg(\lcm(g_1(x),h_1^\dagger(x)))= 58, \deg(\lcm(g_2(x),h_2^\dagger(x)))= 57.\]
		Therefore, from Theorem 3.12 (iv), it follows that \[\dimension_{\mathbb{F}_{5^2}}(\Hull_H(\Phi_M(\mathcal{C})))=m+2n-22-58-57=7.\] Then 
		\[ e=m+2n-k-\dimension_{\mathbb{F}_{5^2}}(\Hull_H(\Phi_M(\mathcal{C})))=1\]
		Hence, by Theorem 4.6, a
		QECC with parameters $[[145, 129, 4]]_5$ is obtained. And a higher dimension and a larger minimum distance are achieved by our code compared with the QECC $[[145, 123, 3]]_5$ from \cite{Tian2024}.
		
	\end{example}

	\begin{example}
		Assume that \( m = 40 \), \( n = 91 \), \( \lambda = \beta = 1 \), \( \mathcal{S} = \mathbb{F}_{9^2} \times (\mathbb{F}_{9^2} + v\mathbb{F}_{9^2})(v^2 = v) \) and \( \omega \) is a generator of \( \mathbb{F}_{9^2}^* \). 
		Take \[ g_0(x) = (x + 1)(x + \omega^2)(x + \omega^4)(x + \omega^6),\] 
		\[g_1(x) = (x + 2)(x^3 + \omega^{10}x^2 + \omega^{10}x + 2), \] \[g_2(x) = (x + 1)(x^3 + \omega^{10}x^2 + \omega^{50}x + 1).\] 
		By Lemma 2.10 and MAGMA, it can be concluded that $d(C_{m})=5$ and $d(\phi_M(\mathscr{C}_{n})) = 5$.
		According to Lemma 2.9, $d(\Phi_M(\mathcal{C}))=5$, and by Lemma 2.10, $k=m+2n-\deg(g_0(x))-\deg(g_1(x))-\deg(g_2(x))=210$.
		
		Then by Theorem 3.15 (ii) and MAGMA, $d(\Sum_H(C_{m}))=4$. Likewise by Theorem 3.18 (ii) and MAGMA, $d(\Sum_H(\phi_M(\mathscr{C}_{n})))=4$. Furthermore, based on Theorem 3.18 (i), \[d(\Sum_H(\Phi_M(\mathcal{C}))) = \min \{d(\Sum_H(C_{m})), d(\Sum_H(\phi_M(\mathscr{C}_{n})))\} =4.\] 
		
		In addition, we use the MAGMA to obtain \[\deg(\lcm(g_0(x),h_0^\dagger(x)))=37,\deg(\lcm(g_1(x),h_1^\dagger(x)))= 88,\deg(\lcm(g_2(x),h_2^\dagger(x)))= 88.\]
		Therefore, from Theorem 3.12 (iv), it follows that \[\dimension_{\mathbb{F}_{9^2}}(\Hull_H(\Phi_M(\mathcal{C})))=m+2n-37-176=9.\] Then 
		\[e=m+2n-k-\dimension_{\mathbb{F}_{9^2}}(\Hull_H(\Phi_M(\mathcal{C})))=3\]
		Hence, applying Theorem 4.6 yields a QECC with parameters $[[225, 201, 5]]_9$. The dimension of our code is superior to that of the QECC $[[225, 193, 5]]_9$ from \cite{Edel2025}.
		
	\end{example}

	Table 1 shows some good QECCs from the Hermitian hulls and sums of sepearable $\mathcal{S}$-$t$-constacyclic codes by Theorem 4.6, where $\mathcal{S} = \mathbb{F}_{p^2} \times (\mathbb{F}_{p^2} +v\mathbb{F}_{p^2})$ and $t=(\lambda, \beta(1-2v))$. In the table, the generator polynomial $g_i(x) = b_0 + b_1x + \cdots + b_k x^k$ is written as $b_0 b_1 \cdots b_k$, and $A$ represents $10$.
	
	\begin{remark}
		(1) Since the lengths of QECCs constructed from constacyclic codes over the single ring are all integer multiples, we use the hulls and sums of separable constacyclic codes over the mixed alphabets to yield QECCs, and can also generate some QECCs of prime length with good parameters, such as $[[193, 179, 3]]_{5}$, $[[109, 99, 3]]_7$, $[[113, 97, 4]]_{7}$, $[[193, 177, 4]]_{7}$, $[[127, 117, 3]]_{9}$.
		
		(2) QECCs $[[193,179,3]]_{5}$, $[[205,191,3]]_{5}$, $[[95,79,4]]_{7}$, $[[225,201,5]]_{9}$ have a higher dimension than the known QECCs $[[193,175,3]]_{5}$, $[[205,187,3]]_{5}$, $[[95,77,4]]_{7}$, $[[225,193,5]]_{9}$ in \cite{Edel2025}.
		Compared with QECCs $[[108,86,4]]_{7}$, $[[189,165,4]]_{11}$ in \cite{Qian2026}, QECCs $[[110,96,4]]_{7}$, $[[187,169,4]]_{11}$ get a higher code rate.
		A higher code rate and a larger minimum distance are achieved by QECCs $[[22,6,5]]_{11}$, $[[62,46,5]]_{11}$ compared with QECCs $[[21,3,4]]_{11}$, $[[63,45,3]]_{11}$ in \cite{Caliskan2023-1}. 
		And a higher dimension and a larger minimum distance are possessed by QECCs $[[145,129,4]]_{5}$, $[[96,86,3]]_{11}$ compared with QECCs $[[145,123,3]]_{5}$ in \cite{Tian2024}, $[[96,84,2]]_{11}$ in \cite{Caliskan2023-1}. 
		\end{remark}

	\begin{landscape}
		\begin{table}[ht]
			\footnotesize \centering
			\caption{QECCs from Hermitian hulls and sums of separable $\mathcal{S}$-$t$-constacyclic codes}
			\label{tab:my_table}
			\begin{tabular}{cccccccccccc}
				
				\hline
				\( p \) & \( (m,n) \) & \( t \) & \( g_0(x) \)& \( g_1(x) \) & \( g_2(x) \) & $\dimension_{\mathbb{F}_{p^2}}(\Hull_H(\Phi_M(\mathcal{C})))$ &  $e$ & \([[n, k, d]]_p\) & \([[n', k', d']]_p\) \\
				\hline

				$5$    & $(30, 28)$    & $(\omega^{12}, 1-2v)$    & $(21)(\omega^{2}1)^2$    & $11$    &$\omega^{3}\omega^{7}\omega^{2}1$     &$6$   &  $1$   &\([[87, 73, \geq 3]]_5\)     & \([[86, 70, 3]]_5\)\cite{Tian2024}      \\
				
				$5$    & $(65, 16)$    & $(\omega^{12}, 1-2v)$    & $(11)^2(1\omega^{7}1)$    & $11$    &$\omega^{3}0001$     &$8$   &  $1$   &\([[98, 80, \geq 3]]_5\)     & \([[98, 78, 3]]_5\)\cite{Tian2024}      \\
				
				$5$    & $(90, 4)$    & $(\omega^{12}, 1-2v)$    & $(\omega^{2}1)^2(\omega^{2}001)$    & $11$    &$\omega^{3}1$     &$6$   &  $1$   &\([[99, 85, \geq 3]]_5\)     & \([[99, 83, 3]]_5\)\cite{Qian2026}      \\
				
				$5$    & $(24, 60)$    & $(1, 1-2v)$    & $(11)(\omega1)(\omega^{2}1)$    & $(11)(\omega^{2}1)$    &$(\omega1)^3$     &$7$   &  $1$   &\([[145, 129, \geq 4]]_5\)     & \([[145, 123, 3]]_5\)\cite{Tian2024}      \\
				
				$5$    & $(24, 60)$    & $(1, 1-2v)$    & $(11)(\omega1)$    & $11$    &$(\omega1)^2$     &$4$   &  $1$   &\([[145, 135, \geq 3]]_5\)     & \([[144, 134, 3]]_5\)\cite{Edel2025}      \\
				
				$5$    & $(40, 60)$    & $(1, 1-2v)$    & $(11)(21)(\omega^{3}1)^3$    & $(11)(\omega^{2}1)$    &$(\omega1)^3$     &$10$   &  $0$   &\([[160, 140, \geq 4]]_5\)     & \([[160, 134, 4]]_5\)\cite{Pandey2024}      \\
				
				$5$    & $(13, 90)$    & $(\omega^{12}, 1-2v)$    & $1\omega^{7}1$    & $\omega^{4}001$    &$(\omega^{2}1)^2$     &$7$   &  $0$   &\([[193, 179, \geq 3]]_5\)     & \([[193, 175, 3]]_5\)\cite{Edel2025}      \\
				
				$5$    & $(24, 90)$    & $(1, 1-2v)$    & $(11)(\omega1)$    & $\omega^{4}001$    &$(\omega^{2}1)^2$     &$6$   &  $1$   &\([[205, 191, \geq 3]]_5\)     & \([[205, 187, 3]]_5\)\cite{Edel2025}      \\

				$7$    & $(70, 12)$    & $(1, 1-2v)$    & $(61)^3(1\omega^{3}1)$    & $(11)(\omega^{4}1)$    &$\omega^{2}1$     &$7$   &  $1$   &\([[95, 79, \geq 4]]_7\)     & \([[95, 77, 4]]_7\)\cite{Edel2025}      \\
				
				$7$    & $(84, 12)$    & $(1, 1-2v)$    & $(11)^2(\omega^{4}1)$    & $11$    &$\omega^{2}1$     &$4$   &  $1$   &\([[109, 99, \geq 3]]_7\)     & \([[108, 93, 3]]_7\)\cite{Edel2025}      \\
				
				$7$    & $(60,24)$    & $(1, 1-2v)$    & $(11)(31)(3\omega1)$    & $(11)(\omega^{2}1)$    &$\omega1$     &$5$   &  $2$   &\([[110, 96, \geq 4]]_7\)     & \([[108, 86, 4]]_7\)\cite{Qian2026}      \\
				
				$7$    & $(80, 16)$    & $(1, 1-2v)$    & $(\omega^{3}1)(\omega^{9}1)(\omega^{12}\omega^{3}1)$    & $(11)(\omega^{3}1)$    &$\omega^{3}01$     &$7$   &  $1$   &\([[113, 97, \geq 4]]_7\)     & \([[112, 92, 4]]_7\)\cite{Edel2025}      \\
				
				$7$    & $(75, 20)$    & $(1, 1-2v)$    & $(31)(1\omega^{2}1)$    & $\omega^{12}1$    &$\omega^{12}\omega^{3}1$     &$5$   &  $1$   &\([[116, 104, \geq 3]]_7\)     & \([[116, 100, 3]]_7\)\cite{Tian2024}      \\
				
				$7$    & $(84, 24)$    & $(1, 1-2v)$    & $(11)^2(\omega^{4}1)$    & $11$    &$\omega1$     &$4$   &  $1$   &\([[133, 123, \geq 3]]_7\)     & \([[132, 118, 3]]_7\)\cite{Edel2025}      \\
				
				$7$    & $(24, 84)$    & $(\omega^{24}, 1-2v)$    & $(\omega1)(\omega^{3}1)(\omega^{5}1)$    & $(11)^3(\omega^{4}1)$    &$\omega^{2}1$     &$7$   &  $1$   &\([[193, 177, \geq 4]]_7\)     & \([[192, 176, 4]]_7\)\cite{Edel2025}      \\

				$9$    & $(60, 32)$    & $(1, 1-2v)$    & $(11)^2(\omega^{4}1)$    & $11$    &$\omega^{5}0001$     &$6$   &  $2$   &\([[126, 110, \geq 3]]_9\)     & \([[126, 20, 3]]_9\)\cite{Ashraf2022}      \\
				
				$9$    & $(5, 60)$    & $(\omega^{40}, 1-2v)$    & $(11)(\omega^{16}1)$    & $(11)^2$    &$\omega^{2}1$     &$3$   &  $2$   &\([[127, 117, \geq 3]]_9\)     & \([[125, 108, 3]]_9\)\cite{Edel2025}      \\
				
				$9$    & $(40, 91)$    & $(1, 1-2v)$    & $(11)(\omega^{2}1)(\omega^{4}1)(\omega^{6}1)$    & $(21)(2\omega^{10}\omega^{10}1)$    &$(11)(1\omega^{50}\omega^{10}1)$     &$9$   &  $3$   &\([[225, 201, \geq 5]]_9\)     & \([[225, 193, 5]]_9\)\cite{Edel2025}      \\

				$11$    & $(11, 5)$    & $(\omega^{60}, 1-2v)$    & $(11)^4$    & $(21)(81)$    &$(11)(41)$     &$7$   &  $1$    &\([[22, 6, \geq 5]]_{11}\)     & \([[21, 3, 4]]_{11}\)\cite{Caliskan2023-1}      \\
				
				$11$    & $(20, 20)$    & $(1, 1-2v)$    & $(11)(\omega^{6}1)(21)(\omega^{18}1)$    & $(11)(\omega^{6}1)$    &$(\omega^{3}1)(\omega^{9}1)$     &$6$   &  $2$    &\([[62, 46, \geq 5]]_{11}\)     & \([[63, 45, 3]]_{11}\)\cite{Caliskan2023-1}      \\
				
				$11$    & $(15, 40)$    & $(1, 1-2v)$    & $(\omega^{4}1)(21)$    & $11$    &$\omega^{3}01$     &$4$   &  $1$    &\([[96, 86, \geq 3]]_{11}\)     & \([[96, 84, 2]]_{11}\)\cite{Caliskan2023-1}      \\
				
				$11$    & $(11, 88)$    & $(1, 1-2v)$    & $(A1)^3$    & $(11)^3(\omega^{15}1)$    &$\omega^{15}01$     &$9$   &  $0$    &\([[187, 169, \geq 4]]_{11}\)     & \([[189, 165, 4]]_{11}\)\cite{Qian2026}      \\

				\hline
			\end{tabular}
		\end{table}
	\end{landscape}

	\section*{5 Conclusion}
	
	In this work, we put forward two methods to produce new QECCs from separable constacyclic codes over $\mathcal{S} = \mathbb{F}_q \times (\mathbb{F}_q + v \mathbb{F}_q) $. First, we discussed the generator polynomials of the Euclidean and Hermitian hulls of separable $\mathcal{S}$-$t$-constacyclic codes and their Gray images, where $t= (\lambda, \beta(1-2v))$. Then, the generator polynomials of the Euclidean and Hermitian sums and their Gray images were also presented. As an application, we derived lots of QECCs with improved parameters.

	\section*{Acknowledgements}
	This research was supported by the National Natural Science Foundation of China (Grant Nos.122
	71137, 12501729) and Research Project on Scientific Research and Practical Innovation of 
	Hefei University (Grant No.2025Yxscx04).
	
	\section*{Author Contributions}
	All authors contributed to the study conception and design. All authors discussed the results, revised and approved the manuscript.
	
	\section*{Declarations}
	
	\subsection*{Conflict of Interests}
	The authors have no competing interests to declare that are relevant to the content of this article.

\end{document}